\documentclass[11pt,a4paper]{article}
\usepackage{authblk}
\usepackage[left=3cm,right=3cm,top=3cm,bottom=3cm]{geometry} 
\usepackage{amsmath,amsthm,amssymb,dsfont,epsfig,latexsym,color,braket} 
\usepackage{indentfirst}
\usepackage{lipsum}
\usepackage{csquotes}
\allowdisplaybreaks[4] 

\usepackage[
    backend=bibtex, 
    natbib=true,
    style=phys,
    biblabel=brackets,
    pageranges=false,
   minalphanames=3,
   maxalphanames=4,
    citestyle=alphabetic,
    sorting=nyt,
    giveninits=true,
    abbreviate=true,
    doi=false,
    url=false,
    isbn=false,
    block=space,
    backref=false,
    backrefstyle=two,
]{biblatex}

\DeclareFieldFormat{titlecase}{#1}

\DeclareFieldFormat{bibentrysetcount}{\mkbibparens{\mknumalph{#1}}}
\DeclareFieldFormat{labelalphawidth}{\mkbibbrackets{#1}}
\DeclareFieldFormat{shorthandwidth}{\mkbibbrackets{#1}}

\defbibenvironment{bibliography}
  {\list
     {\printtext[labelalphawidth]{%
        \printfield{labelprefix}%
        \printfield{labelalpha}%
        \printfield{extraalpha}}}
     {\setlength{\labelwidth}{\labelalphawidth}%
      \setlength{\leftmargin}{\labelwidth}%
      \setlength{\labelsep}{\biblabelsep}%
      \addtolength{\leftmargin}{\labelsep}%
      \setlength{\itemsep}{\bibitemsep}%
      \setlength{\parsep}{\bibparsep}}%
      }
  {\endlist}
  {\item}

\usepackage{caption}
\usepackage{subcaption}

\usepackage{tikz}
    \usetikzlibrary{decorations.pathreplacing, arrows, patterns, angles, quotes}
    \usetikzlibrary{shapes.geometric}
    \usetikzlibrary{plotmarks}
    \usetikzlibrary{calc}
    \usetikzlibrary{fadings}
\usepackage{pgfplots}
\usepgfplotslibrary{colormaps}
\usepgfplotslibrary{patchplots}

\pgfplotsset{
	compat=newest
}

\usepackage[bookmarks=true,colorlinks,linkcolor=blue,urlcolor=blue,citecolor=blue,breaklinks]{hyperref}

\newcommand{\deriv}{\mathop{}\mathopen{}\mathrm{d}}

\newcommand{\be}[1]{\begin{equation}\label{#1}}
\newcommand{\ee}{\end{equation}}
\newcommand{\bu}{\bar u}
\newcommand{\bv}{\bar v}

\newcommand\CC{\mathbb C}
\newcommand\NN{\mathbb N}

\newcommand{\ben}{\begin{eqnarray}}
\newcommand{\een}{\end{eqnarray}}

\newtheorem{remark}{Remark}[section]
\newtheorem{remarks}{Remarks}[section]

\newtheorem{examples}{Examples}[section]

\newtheorem{prop}{Proposition}[section]
\newtheorem{Def}{Definition}[section]

\def\<{\langle}
\def\>{\rangle}

\def\ccc{\mathbb{C}}

\numberwithin{equation}{section}

\begin{document}

\date{}

\title{Modified rational six vertex model on a rectangular lattice : new formula, homogeneous and thermodynamic limits}

\author[1]{Matthieu Cornillault \thanks{matthieu.cornillault@univ-tours.fr}}
\author[1]{Samuel Belliard \thanks{samuel.belliard@univ-tours.fr}}

\affil[1]{Institut Denis Poisson, CNRS-UMR 7013, Universit\'e de Tours, Parc de Grandmont, 37200 Tours, France}

\maketitle

\begin{abstract}
We continue the work of \cite{Belliard2024} studying the modified rational six vertex model. We find another formula of the partition function for the inhomogeneous model, in terms of a determinant that mix the modified Izergin one and a Vandermonde one.
This expression enables us to compute the partition function in the homogeneous limit for the rectangular lattice, and then to study the thermodynamic limit. It leads to a new result, we obtain the first order of free energy with boundary effects in the thermodynamic limit.
\end{abstract}

\newpage

\tableofcontents


\newpage

\section*{Introduction}

\addcontentsline{toc}{section}{Introduction}

Following on from the work of one of the authors and collaborator's \cite{Belliard2024}, we study the XXX six vertex model on a rectangular lattice with general boundary conditions (GBC), called the modified rational six vertex (MR6V) model. 

For some boundary conditions, the partition function of the corresponding six-vertex model is represented by a determinant. In the case of XXX or XXZ models on the square lattice with domain wall boundary conditions (DWBC), it is the so-called {\bf Izergin determinant} \cite{Izergin1987,Izergin1992,Korepin1982,Korepin1993}. More complex boundary conditions exist : as example we quote ones related to reflection algebra, which lead to the {\bf Tsuchiya determinant} \cite{Tsuchiya1998}.

A first generalization to rectangular lattice was the so called six vertex model with partial domain wall boundary conditions (pDWBC), introduced in \cite{Foda2012}. 
This model enabled to recover the \textbf{Kostov's determinant} \cite{Kostov2012, Kostov2012a}, linked with the $\mathcal{N} = 4$ supersymmetric Yang-Mills theory.
This model combines boundaries with same spin (up or down) and others that mixed up and down spins.

The GBC are a generalization of the pDWBC : the boundaries no longer contain up or down spins but linear combinations of up and down spins. Related to quantum algebra, the MR6V model is connected to the XXX $\frac12$-spin chain with a general twist, which can be represented by \textbf{modified Bethe vectors}, constructed from the \textbf{modified algebraic Bethe ansatz} \cite{Belliard2015,Belliard2018}. Its partition function, linked with the modified Bethe vectors, is represented by the {\bf modified Izergin determinant} (MID). This determinant is similar to the Izergin determinant but it depends also on an additional parameter and a normalization factor.

In this paper, we consider in particular the thermodynamic limit of the MR6V model. This limit is studied in two cases : first, when the lattice tends to a very long lattice, as in \cite{Minin2021} with a rectangular lattice with pDWBC, and then when it tends to an infinite lattice in both directions, as in \cite{Korepin2000} with a square lattice with DWBC. 

This paper is organized as follows. In the section \ref{MR6V}, after recalling the MR6V model, the definition and the known formulae of its partition function, we present a new formula for this partition function. 
This formula is proved in the appendix \ref{appendixA}, first using a similar reasoning as in \cite{Foda2012} and then, in an alternative proof, showing that the new formulae is equal to known ones. 
Moreover, in the appendix \ref{pDWBC}, we make the link with the pDWBC partition function introduced in \cite{Foda2012}.
Then, in the section \ref{applications}, we use the new formula of the partition function to compute first the homogeneous limit of the MR6V model and, then, its thermodynamic limit. The appendix \ref{appendixC} explains in details the computations for the infinite lattice.
Finally, the appendix \ref{mathematicaltools} presents some useful mathematical tools for our computations.
 
\newpage

\paragraph{Notation} Like in \cite{Belliard2024}, we use a shorthand notation for sets of variables and products over them.
For example, we denote a set of $n$ variables $u_i$ by $\bar u = \{u_1,\dots,u_n\}$. We usually leave the cardinality implicit and note it as $\# \bar u=n$.
The removal of the $i$-th element of the set $\bar u$ is denoted $\bar u_i=\bar u \backslash u_i$. For the product of a two variable function $g(u,v)$ over the set $\bar u$ we use,
\ben\label{nota1}
g(z,\bar u) = \prod_{x\in \bar u}g(z,x),\quad
g(\bar u,z) = \prod_{x\in \bar u}g(x,z),\quad
g(\bar u,\bar v) = \prod_{x\in \bar u,y\in \bar v}g(x,y)\,.
\een
We will also use such notation for the product of commuting operators
\ben\label{nota1o}
B(\bar u)= \prod_{x\in \bar u} B(x).
\een
If no product is involved, a vertical bar is used to indicate the
multivariable function dependency, {\it e.g.},
\ben\label{nota2}
s(u|\bar u) = s(u|{u_1,\dots,u_{n}})\,,\quad r(\bar u|\bar v) = r(u_1,\dots, u_n|v_1,\dots,v_m).
\een
The following functions will be used 
\ben\label{funcs}
g(u,v)=\frac{c}{u-v},\quad f(u,v)=\frac{u-v+c}{u-v},\quad h(u,v)=\frac{u-v+c}{c}.
\een
which satisfy the relations
\ben\label{prop-funcs}
f(u,v)=g(u,v)+1=g(u,v)h(u,v).
\een
We consider also two functions 
\begin{equation}
    \phi_\beta(x) = \frac{c}{x}-\beta\frac{c}{x+c} \quad  \text{and} \quad \psi_k(x) = \left(-\frac{x}{c}\right)^k.
\end{equation}
and we denote $\phi_1 = \phi_{\beta = 1}$.
We denote also the Vandermonde determinants
\begin{equation}
    \Delta(\bu) = \prod_{1 \leqslant i < j \leqslant n} g(u_j,u_i)^{-1} \quad \text{and} \quad \Delta'(\bu) = \Delta(-\bu) = \prod_{1 \leqslant i < j \leqslant n} g(u_i,u_j)^{-1}.
\end{equation}

\newpage

\section{Modified rational six vertex (MR6V) model}\label{MR6V}

\subsection{Definition}

First, we recall the description of the inhomogeneous six-vertex model on the $n \times m$ rectangular lattice in the framework of the Quantum Inverse Scattering Method (QISM).
According to the figure \ref{fig:reseau_inhomogene}, we assign to each line and to each column a vector space, $V_{a_i}$ for the lines $i \in [\![1,n]\!]$ and $V_{b_j}$ for columns $j \in [\![1,m]\!]$, and a free parameter, $u_i$ for rows and $v_j$ for columns, also called inhomogeneity. We order the lines from bottom to top and orient them from right to left. And we order the columns from right to left and orient them from bottom to top.

We assign to each segment of the lattice a spin which can be up or down. Each spin can have only two possible states so the vector spaces $V_{a_i}$ and $V_{b_j}$ are isomorphic to $\ccc^2$.
Each vertex is labeled by a pair of variable $(u,v)$, linked with a pair of vector spaces $(V_a,V_b)$ and is described by a $R$-matrix $R_{ab}(u,v)$ (see the figure \ref{fig:matrice_R}) whose entries are the statistical weights of the model. These entries are determined using the orientation of the line and the column.

\begin{figure}[ht]
    \centering
    \begin{subfigure}[h]{0.5\textwidth}
        \centering
        \includegraphics[width=\textwidth]{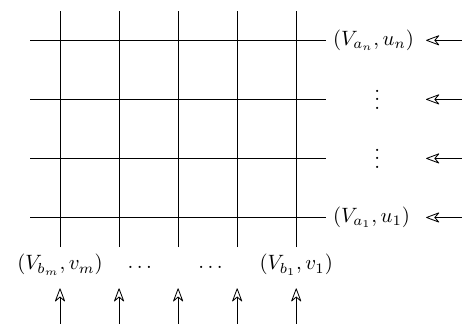}
        \caption{\centering The inhomogeneous six-vertex model on the $n \times m$ rectangular lattice.}
        \label{fig:reseau_inhomogene}
    \end{subfigure}
    \hspace{10mm}
    \begin{subfigure}[h]{0.35\textwidth}
        \centering
        \includegraphics[width=\textwidth]{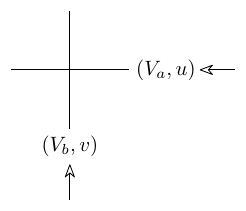}
        \caption{\centering Graphical picture of a vertex and the $R$-matrix.}
        \label{fig:matrice_R}
    \end{subfigure}   
    \caption{\centering Modelisation of the lattice and vertices. The open arrows indicate the orientation of the lines.}
    \label{fig:modelisation_reseau}
\end{figure}

In the case of the six-vertex model, the vertices of the lattice can have only six possible forms (see the figure \ref{fig:six_vertex}). We denote the energy of each possible vertex $\epsilon_i$, $i \in [\![1,6]\!]$, and their statistical weights $\omega_i = \exp(-\epsilon_i/(k_BT))$. Considering the zero field assumption (see \cite{Baxter1982} for details), we have
\begin{equation}
    a := \omega_1 = \omega_2, \quad b := \omega_3 = \omega_4 \quad \text{and} \quad c := \omega_5 = \omega_6.
\end{equation}

\begin{figure}[ht]
    \centering
    \includegraphics[width=\textwidth]{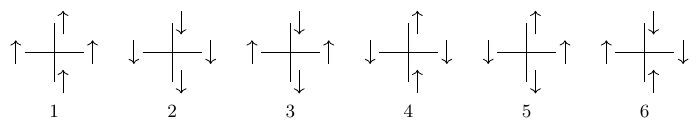}
    \caption{\centering The six different vertices in the QISM.}
    \label{fig:six_vertex}
\end{figure}

For the rational case, using the inhomogeneities $u$ and $v$, the statistical weight of the vertex labeled by $(u,v)$ can be equal to
\begin{equation}
    a(u,v) = \frac{u-v+c}{c}, \quad b(u,v) = \frac{u-v}{c} \quad \text{or} \quad c(u,v) = 1.
\end{equation}
The corresponding $R$-matrix is given by the formula
\be{RXXX}
    R_{ab}(u,v) = R_{ab}(u-v) = \frac{u-v}{c}I + P_{ab}
\end{equation}
where $I$ is the identity matrix and $P_{ab}$ the permutation matrix. $R_{ab}(u,v)$ is an operator acting on $V_a \otimes V_b$ so we can represent it by a $4 \times 4$ matrix :
\begin{equation}
    R_{ab}(u,v) = \begin{pmatrix}
                    \frac{u-v+c}{c} & 0 & 0 & 0 \\
                    0 & \frac{u-v}{c} & 1 & 0 \\
                    0 & 1 & \frac{u-v}{c} & 0 \\
                    0 & 0 & 0 &\frac{u-v+c}{c}
                  \end{pmatrix}.
\end{equation}
It is one of the simplest solution of the Yang-Baxter equation
\begin{equation}\label{eqYB}
    R_{ab}(u,v)R_{ac}(u,w)R_{bc}(v,w) = R_{bc}(v,w)R_{ac}(u,w)R_{ab}(u,v)
\end{equation}
that ensures the integrability of the model (cf definition 7.5.2 in \cite{Chari2000}). 
This $R$-matrix has the property 
\begin{equation}
    [R_{ab}(u,v), B_aB_b] = [R_{ab}(u,v), \hat C_a \hat C_b] = 0
\end{equation}
with $B$ and $\hat C$ are $2 \times 2$ complex matrices. This property has motivated the following definition of the twisted inhomogeneous partition function.

\begin{Def}
    \label{defZ1}
    The partition function of the inhomogeneous six vertex model on the $n \times m$ rectangular lattice with general boundary conditions is
    \begin{equation}
        Z_{nm}(\bu|\bv|B|\hat C) = \underset{\bar a, \bar b}{\mathrm{tr}}\left( \prod_{i=1}^n B_{a_i} \prod_{j=1}^m \hat C_{b_j} \prod_{i=1}^n \prod_{j=1}^m R_{a_ib_j}(u_i-v_j) \right).
    \end{equation}
\end{Def}

\begin{remark}
    The first subscript of $Z_{nm}(\bu|\bv|B|\hat C)$ indicates the cardinal of $\bu$ and the second one the cardinal of $\bv$. It will be useful in the following part.
    
    We recall that $Z_{nm}(\bu|\bv|B|\hat C)$ is symmetric on the elements of $\bu$ and on ones of $\bv$.
\end{remark}

Depending on the nature of the twists $\{B,\hat C\}$, we have different boundary conditions for the partition function. If both $B$ and $\hat C$ are invertible, the lattice is a torus. If only one of the twists is invertible, we have a cylinder. Finally, if $B$ and $\hat C$ are not invertible, we have a plane. From now on, we restrict ourselves on this last case in this paper.

We consider a plane with general boundary conditions (see the figure \ref{fig:bords_generaux}), defined by the four vectors
\begin{equation}
\ket{x} = x_1 \ket{1} + x_2 \ket{2} = \begin{pmatrix} x_1 \\ x_2 \end{pmatrix}
\end{equation}
with $x \in \{n,s,e,w\}$ some free set of complex vectors. $(\ket{1},\ket{2})$ forms the canonical basis of $\ccc^2$, the first vector representing the up spin and the second one the down spin. We denote the dual vector as $\bra{x}$ and we have $\braket{i|j} = \delta_{ij}$ for $i,j \in \{1,2\}$.

\begin{Def}
    \label{defZ2}
    We define the following twist
    \begin{equation}
        B = \ket{e} \otimes \bra{w} \quad \text{and} \quad \hat C = \ket{s} \otimes \bra{n}.
    \end{equation}
    We can rewrite the partition function with the following formula
    \begin{equation}
        Z_{nm}(\bu|\bv|B|\hat C) = \bra{W} \otimes \bra{N} \left( \prod_{i=1}^n \prod_{j=1}^m R_{a_ib_j}(u_i-v_j) \right) \ket{E} \otimes \ket{S}
    \end{equation}
    with 
    \begin{equation}
        \bra{N} = \bigotimes_{j=1}^m \bra{n}_{b_j}, \quad \bra{W} = \bigotimes_{i=1}^n \bra{w}_{a_i}, \quad \ket{S} = \bigotimes_{j=1}^m \ket{s}_{b_j} \quad \text{and} \quad \ket{E} = \bigotimes_{i=1}^n \ket{e}_{a_i}.
    \end{equation}
\end{Def}

\begin{figure}[ht]
    \centering
    \includegraphics[width=0.4\textwidth]{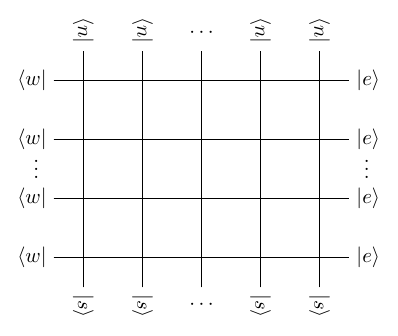}
    \caption{\centering The $n \times m$ rectangular lattice with general boundary conditions.}
    \label{fig:bords_generaux}
\end{figure}

\subsection{Known results}

In \cite{Belliard2024}, using a linear system approach, the authors constructed a formula for the partition function with the modified Izergin determinant $K_{nm}^{(\beta)}(\bu|\bv)$ (see \cite{Belliard2018} for more details).

\begin{prop}
    \cite{Belliard2024} The partition function of the inhomogeneous rational six-vertex model with general boundary conditions is 
    \begin{equation}\label{ZwithK}
        Z_{nm}(\bu|\bv|B|\hat C) = \frac{\mathrm{tr}(B)^n\mathrm{tr}(\hat C)^m}{(1-\beta)^m} \frac{1}{g(\bu,\bv)}K_{nm}^{(\beta)}(\bu|\bv) 
    \end{equation}
    with 
    \begin{equation}\label{defbeta}
        \beta = 1-\frac{\mathrm{tr}(B)\mathrm{tr}(\hat C)}{\mathrm{tr}(B \hat C)}, \quad \mathrm{tr}(B) = \braket{w|e}, \quad \mathrm{tr}(\hat C) = \braket{n|s} \quad \text{and} \quad \mathrm{tr}(B \hat C) = \braket{n|e}\braket{w|s}.
    \end{equation}
    The modified Izergin determinant can be expressed with two formulae \cite{Belliard2018} 
    \begin{align}
        K_{nm}^{(\beta)}(\bu|\bv) &= \underset{1 \leqslant k,l \leqslant m}{\mathrm{det}}\left(-\beta \delta_{kl} + \frac{f(\bu,v_k)f(v_k,\bv_k)}{h(v_k,v_l)}\right) \label{K1} \\ 
        &= (1-\beta)^{m-n} \underset{1 \leqslant k,l \leqslant n}{\mathrm{det}}\left( f(u_k,\bv) \delta_{ij} -\beta \frac{f(u_k,\bu_k)}{h(u_k,u_l)}\right). \label{K2}
    \end{align}
    For the special case of a square lattice ($m = n$), we have an additional formula \cite{Belliard2024}
    \begin{equation}\label{K3}
        K_n^{(\beta)}(\bu|\bv) = \frac{h(\bu,\bv)}{\Delta(\bu)\Delta'(\bv)} \underset{1 \leqslant i,j \leqslant n}{\mathrm{det}} \left( \phi_\beta(u_i-v_j) \right).
    \end{equation}
\end{prop}

\begin{remark}
    For the specific values of the parameters $m=n$, $\bra{n} = \bra{2}$, $\bra{w} = \bra{1}$, $\ket{e} = \ket{2}$ and $\ket{s} = \ket{1}$ corresponding to domain wall boundary conditions, we have $\mathrm{tr}(B) = \mathrm{tr}(\hat C) = 0$ and $\mathrm{tr}(B\hat C) = 1$. It follows that $\beta = 1$ and one recover the usual Izergin determinant. 
\end{remark}

\begin{remark}
   Here, the same twist is used for all rows and the same for all columns. It is possible to define local twists $B_i$ and $\hat{C}_j$ (not proportional on a same boundary, else we return to our original problem).
    In that case, the formulation in terms of a single family of operators — which is crucial in our calculations of the modified Izergin determinant — fails. Such models require further research, especially if one seeks a determinant formula. This point remains unclear to us.
    
    For certain configurations involving local twists, it is the partition function of the model comprising all possible boundary conditions (meeting a certain criterion) that can be described by a determinant formula, as is the case with the pDWBC model. Indeed, we observe that, across all the configurations considered, there is symmetry in the boundary conditions that appear.
\end{remark}


\subsection{A new expression for rectangular lattice}

Using the method developed by Foda and Wheeler in \cite{Foda2012} for the partial domain wall boundary conditions, we found a new expression for the partition function for the $n \times m$ rectangular lattice that generalizes the formula \eqref{ZwithK} with \eqref{K3}.

We start from the formula \eqref{ZwithK} with \eqref{K3} considering a square lattice whose the length of the side is $\max(n,m)$. Then, we remove the extra parameters doing them tend toward infinity to obtain the desired width, so $\min(n,m)$. This operation consists in remove lines or columns in the lattice, due to the fact that 
\begin{equation}
 \frac{c}{u}R(u,v) \underset{u \rightarrow \infty}{\rightarrow} I \;\;\; \text{and} \;\;\; \frac{c}{-v}R(u,v) \underset{v \rightarrow \infty}{\rightarrow} I.
\end{equation}
Doing this for all the extra parameters, we can express the partition function for the rectangular lattice with the partition function for the larger square lattice (using also the definition \ref{defZ1} of the partition function).

\begin{prop}\label{defZt}
For $n \leqslant m$, 
\begin{equation}
    Z_{nm}(\bu|\bv|B|\hat C) = \lim\limits_{u_m, \dots , u_{n+1} \rightarrow \infty} \left( \frac{1}{\mathrm{tr}(B)^{m-n}} \prod_{i=n+1}^m \left(\frac{c}{u_i}\right)^m Z_{mm}(\bu\vert\bv|B|\hat{C}) \right).
\end{equation}
and for $n \geqslant m$,
\begin{equation}
    Z_{nm}(\bu|\bv|B|\hat C) = \lim\limits_{v_n, \dots , v_{m+1} \rightarrow \infty} \left( \frac{1}{\mathrm{tr}(\hat C)^{n-m}} \prod_{j=m+1}^n \left(\frac{c}{-v_j}\right)^n Z_{nn}(\bu\vert\bv|B|\hat{C}) \right).
\end{equation}
The limits are to be taken sequentially, starting respectively with $u_m$ and $v_n$.
\end{prop} 

\begin{prop}\label{propnewZ}
Starting from the previous proposition and following the steps in section 2 of \cite{Foda2012}, we obtained a new explicit formula for $Z_{nm}(\bu|\bv|B|\hat C)$ in terms of the determinant of a block matrix :
    \begin{align}
        Z_{nm}(\bu|\bv|B|\hat C) = &\frac{\mathrm{tr}(B)^n\mathrm{tr}(\hat C)^m}{(1-\beta)^{\min(n,m)}} \frac{h(\bu,\bv)}{g(\bu,\bv)\Delta(\bu)\Delta'(\bv)} \notag \\ 
        & \times \left\{ 
            \begin{array}{cc}
                 \underset{1 \leqslant i,j \leqslant m}{\mathrm{det}} \left( \left\{ 
                 \begin{array}{cc}
                    \phi_\beta(u_i-v_j) & \text{if} \;\; i \leqslant n \\
                    \psi_{m-i}(-v_j) & \text{if} \;\; i \geqslant n+1
                \end{array} \right. \right) & \text{if} \;\; n \leqslant m \\
                \underset{1 \leqslant i,j \leqslant n}{\mathrm{det}} \left( \left\{ 
                \begin{array}{cc}
                    \phi_\beta(u_i-v_j) & \text{if} \;\; j \leqslant m \\
                    \psi_{n-j}(u_i) & \text{if} \;\; j \geqslant m+1
                \end{array} \right. \right) & \text{if} \;\; n \geqslant m
            \end{array} \right. . \label{newformulaZ}
    \end{align}
\end{prop}

\begin{remark}     
    The previous formulae are indeed a generalization of the formula \eqref{ZwithK} with \eqref{K3} : they are $\max(n,m) \times \max(n,m)$ determinant where the first $\min(n,m)$ rows or columns are of the modified Izergin determinant-type and the other are of Vandermonde determinant type.
    Moreover, we obtain a new expression for the modified Izergin determinant $K_{nm}^{(\beta)}(\bu|\bv)$.
\end{remark}

\section{Applications}\label{applications}

\subsection{Homogeneous limit}

The previous part deals with the inhomogeneous lattice. In this section, we consider the special case of the homogeneous lattice : we consider therefore the limits $u_i \rightarrow u$ and $v_j \rightarrow v$ for all $i \in [\![1,n]\!]$ and $j \in [\![1,m]\!]$. So the vertices will be parametrized by $x = u-v$. We will follow the construction made by Izergin, Coker and Korepin in \cite{Izergin1992} (the result of the computation was first written by Izergin in \cite{Izergin1987}).

We start from the formula \eqref{newformulaZ}. Doing Taylor expansions of the functions $\phi$ and $\psi$, exactly as in the section 7 of \cite{Izergin1992}, we obtain the following formulae
\begin{align}
    Z_{nm}(u|v|B|\hat C) &= \frac{\mathrm{tr}(B)^n\mathrm{tr}(\hat C)^m}{(1-\beta)^{\min(n,m)}} \frac{c^{\frac{n(n-1)+m(m-1)}{2}}}{\phi_1(u-v)^{nm}\prod_{k=1}^n(k-1)!\prod_{l=1}^m(l-1)!} \notag \\
    &\times \left\{ \begin{array}{cc}
                \underset{1 \leqslant i,j \leqslant m}{\mathrm{det}} \left( \left\{ \begin{array}{cc}
               \phi^{(i+j-2)}_\beta(u-v) & \text{if} \;\; i \leqslant n \\
                \psi^{(j-1)}_{m-i}(-v) & \text{if} \;\; i \geqslant n+1
           \end{array} \right. \right) & \text{for} \;\; n \leqslant m \\
                 \underset{1 \leqslant i,j \leqslant n}{\mathrm{det}} \left( \left\{ \begin{array}{cc}
               \phi^{(i+j-2)}_\beta(u-v) & \text{if} \;\; j \leqslant m \\
                \psi^{(i-1)}_{n-j}(u) & \text{if} \;\; j \geqslant m+1
           \end{array} \right. \right) & \text{for} \;\; n \geqslant m
            \end{array}  
        \right. .
\end{align}    
However, because $\psi_k$ is a polynomial function of degree $k$, the previous matrices contain a lot of zeros and the partition function is
\begin{itemize}
    \item for $n \leqslant m$,
    \begin{align}
        Z_{nm}(u&|v|B|\hat C) = \frac{\mathrm{tr}(B)^n\mathrm{tr}(\hat C)^m}{(1-\beta)^{n}} \frac{c^{\frac{n(n-1)+m(m-1)}{2}}}{\phi_1(u-v)^{nm}\prod_{k=1}^n(k-1)!\prod_{l=1}^m(l-1)!} \notag \\
        &\times \underset{1 \leqslant i,j \leqslant m}{\mathrm{det}} \left( \left\{ 
            \begin{array}{cc}
               \phi^{(i+j-2)}_\beta(u-v) & \text{if} \;\; i \leqslant n \\
                \psi^{(j-1)}_{m-i}(-v) & \text{if} \;\; i \geqslant n+1 \;\; \text{and} \;\; i+j < m+1 \\
                (-1)^{j-1} \frac{(j-1)!}{c^{j-1}} & \text{if} \;\; i \geqslant n+1 \;\; \text{and} \;\; i+j = m+1 \\
                0 & \text{if} \;\; i \geqslant n+1 \;\; \text{and} \;\; i+j > m+1
            \end{array} \right. \right),
    \end{align}
    \item for $n \geqslant m$,
    \begin{align}
        Z_{nm}(u&|v|B|\hat C) = \frac{\mathrm{tr}(B)^n\mathrm{tr}(\hat C)^m}{(1-\beta)^{\min(n,m)}} \frac{c^{\frac{n(n-1)+m(m-1)}{2}}}{\phi_1(u-v)^{nm}\prod_{k=1}^n(k-1)!\prod_{l=1}^m(l-1)!} \notag \\
        &\times \underset{1 \leqslant i,j \leqslant n}{\mathrm{det}} \left( \left\{ 
            \begin{array}{cc}
               \phi^{(i+j-2)}_\beta(u-v) & \text{if} \;\; j \leqslant m \\
                \psi^{(i-1)}_{n-j}(u) & \text{if} \;\; j \geqslant m+1 \;\; \text{and} \;\; i+j < n+1 \\
                (-1)^{i-1} \frac{(i-1)!}{c^{i-1}} & \text{if} \;\; j \geqslant m+1 \;\; \text{and} \;\; i+j = n+1 \\
                0 & \text{if} \;\; j \geqslant m+1 \;\; \text{and} \;\; i+j > n+1 \\
            \end{array} \right. \right).
    \end{align}
\end{itemize}
With successive developments of the determinants (along its last row for the case $n \leqslant m$ and along its last column for the case $n \geqslant m$), we obtain therefore
\begin{itemize}
    \item for $n \leqslant m$,
    \begin{align}
        Z_{nm}(x|B|\hat C) = \frac{\mathrm{tr}(B)^n\mathrm{tr}(\hat C)^m}{(1-\beta)^n} &\frac{(-1)^{(m-n)(m+1)}}{\phi_1(x)^{nm}}\frac{c^{n(m-1)}}{\prod_{k=1}^n (k-1)!(m-n+k-1)!} \notag \\ 
        & \times \underset{1\leqslant i,j \leqslant n}{\mathrm{det}} \left( \phi_\beta^{(m-n+i+j-2)}(x) \right) \label{Zhomo1}
    \end{align}
    \item and for $n \geqslant m$,
    \begin{align}
        Z_{nm}(x|B|\hat C) = \frac{\mathrm{tr}(B)^n\mathrm{tr}(\hat C)^m}{(1-\beta)^m} & \frac{(-1)^{(n-m)(n+1)}}{\phi_1(x)^{nm}}\frac{c^{m(n-1)}}{\prod_{k=1}^m (k-1)!(n-m+k-1)!} \notag \\ 
        & \times \underset{1\leqslant i,j \leqslant m}{\mathrm{det}} \left( \phi_\beta^{(n-m+i+j-2)}(x) \right). \label{Zhomo2}
    \end{align}
\end{itemize}
These expressions depend only on the parameter $x = u-v$, which is consistent with the configuration of the vertices. 
Next, we can compute the derivative of $\phi_\beta$ 
\begin{align}
 \phi_\beta^{(k)}(x) &= (-1)^k k! \left( \frac{c}{x^{k+1}} - \beta \frac{c}{(x+c)^{k+1}} \right) \\
 &= (-1)^k k! \frac{\phi_1(x)^{k+1}}{c^{k}} \left( \left(1+\frac{x}{c}\right)^{k+1} - \beta \left(\frac{x}{c}\right)^{k+1} \right)
\end{align}
and we obtain the following expression
\begin{align}
    Z_{nm}&(x|B|\hat C) = \frac{\mathrm{tr}(B)^n\mathrm{tr}(\hat C)^m}{(1-\beta)^{\min(n,m)}} \notag \\
    & \times \underset{1\leqslant i,j \leqslant \min(n,m)}{\mathrm{det}} \left( \binom{d+i+j-2}{i-1} \left( \left(1+\frac{x}{c}\right)^{d+i+j-1} - \beta \left(\frac{x}{c}\right)^{d+i+j-1} \right) \right). \label{Zhomofin}
\end{align}
with $d = |n-m|$.

\begin{Def}
    We denote $Z_{nm}^0(B|\hat C)$ the value of the partition function for $x = 0$, \textit{i.e.} when the lattice is completely homogeneous. Using the proposition \ref{propdet} and factorizing the $c$ out of the determinant, we have
\begin{equation}
    Z_{nm}^0(B|\hat C) := \frac{\mathrm{tr}(B)^n\mathrm{tr}(\hat C)^m}{(1-\beta)^{\min(n,m)}} = \mathrm{tr}(B\hat C)^{\min(n,m)} \left\{ \begin{array}{cc}
         \mathrm{tr}(\hat C)^d & \text{if} \; \; n \leqslant m  \\
         \mathrm{tr}(B)^d & \text{if} \; \; n \geqslant m 
    \end{array}\right.. 
\end{equation}
\end{Def}

\subsection{Physical interpretation for completely homogeneous finite lattice}

To have a global physical interpretation, we consider the statistical weights non-dimensionaless (only in this subsection)
\begin{equation}
 a(x) = x+c, \qquad b(x) = x \qquad \text{and} \qquad c(x) = c = \exp\left(-\frac{\epsilon}{k_BT}\right).
\end{equation}
With these statistical weights, $Z_{nm}(0|B|\hat C) = Z_{nm}^0(B|\hat C) c^{nm}$.
Assuming that we are in the canonical ensemble, the total free energy of the completely homogeneous lattice is $F_{nm}^{tot} = -k_BT\ln (Z_{nm}(0\vert B \vert \hat C))$ so
\begin{equation}
F_{nm}^{tot} = nm \epsilon - k_BT \left(\min(n,m) \ln(\mathrm{tr}(B\hat C)) + d \left\{ 
 \begin{array}{cc}
    \ln(\mathrm{tr}(\hat C)) & \text{if} \;\; n \leqslant m \\
    \ln(\mathrm{tr}(B)) & \text{if} \;\; n \geqslant m 
 \end{array}\right. \right).
\end{equation}
We see bulk and boundary effects in the expression of $F_{nm}^{tot}$ because $B$ and $\hat C$ are directly linked to the boundaries. Moreover, the boundary effects are linear in the number of rows and columns, contrary to the bulk effects which are quadratic. 

Assuming moreover that the boundaries do not depend on the temperature $T$, we can compute others physical characteristics.
The average energy $\bar{E}^c_{nm}$, the fluctuation of the energy $\Delta(E^c_{nm})$, the heat capacity at constant volume $C_{V,nm}$ and the canonical entropy $S^c_{nm}$ are respectively

\begin{equation}
 \bar{E}_{nm}^c = F_{nm}^{tot} - T \frac{\partial F_{nm}^{tot}}{\partial T} = nm \epsilon, \quad \Delta(E^c_{nm})^2 = -k_BT^3 \frac{\partial^2 F_{nm}^{tot}}{\partial T^2} = 0,
\end{equation}

\begin{equation}
 C_{V,nm} = - T \frac{\partial F_{nm}^{tot}}{\partial T} = k_B T \left( \min(n,m) \ln(\mathrm{tr}(B\hat C)) + d\left\{ 
 \begin{array}{cc}
    \ln(\mathrm{tr}(\hat C)) & \text{if} \;\; n \leqslant m \\
    \ln(\mathrm{tr}(B)) & \text{if} \;\; n \geqslant m                 
 \end{array}\right. \right),
\end{equation}
\begin{equation}
 S^c_{nm} = - \frac{\partial F_{nm}^{tot}}{\partial T} = k_B \left( \min(n,m) \ln(\mathrm{tr}(B\hat C)) + d\left\{ 
 \begin{array}{cc}
    \ln(\mathrm{tr}(\hat C)) & \text{if} \;\; n \leqslant m \\
    \ln(\mathrm{tr}(B)) & \text{if} \;\; n \geqslant m            
 \end{array}\right. \right).
\end{equation}
In the average energy, only the bulk effects appear. However, in the heat capacity and the canonical entropy, only the boundary effects appear. 


\subsection{Thermodynamic limit}

For a rectangular lattice, we can consider two different thermodynamic limits : in the first part, when $\max(n,m) \rightarrow +\infty$ and, in the second part, when $n,m \rightarrow +\infty$ where $d = |n-m|$ remains constant. We consider that $c,x > 0$ (as Boltzmann weights) in the following part.

\subsubsection{The semi-infinite lattice}

This question was treated by Pronko, Pronko and Minin in \cite{Pronko2019,Minin2021} for the partial domain wall boundary conditions. 
We consider the limit $\max(n,m) \rightarrow +\infty$ (so $d \rightarrow +\infty$ but $\min(n,m)$ is constant).
In the formula \eqref{Zhomofin}, we factorize the term $\left(1+\frac{x}{c}\right)^{d+i+j-1}$ out of the determinant and we obtain 
\begin{align}
    Z_{nm}(x|B|\hat C) = &Z_{nm}^0(B|\hat C) \left(1+\frac{x}{c}\right)^{nm} \notag \\ & \times \underset{1\leqslant i,j \leqslant \min(n,m)}{\mathrm{det}} \left( \binom{d+i+j-2}{i-1} \left( 1 - \beta \left(\frac{x}{x+c}\right)^{d+i+j-1} \right) \right). \label{Zhomofactorize}
\end{align}
Because $x,c >0$, $\frac{x}{x+c} < 1$ and $\left(\frac{x}{x+c}\right)^{d+i+j-1} = \mathcal{O}\left(\left(\frac{x}{x+c}\right)^{d}\right)$. So 
\begin{align}
    \underset{1\leqslant i,j \leqslant \min(n,m)}{\mathrm{det}} &\left( \binom{d+i+j-2}{i-1} \left( 1 - \beta \left(\frac{x}{x+c}\right)^{d+i+j-1} \right) \right) \notag \\
    &=  \underset{1\leqslant i,j \leqslant \min(n,m)}{\mathrm{det}} \left( \binom{d+i+j-2}{i-1} \left( 1 + \mathcal{O}\left(\left(\frac{x}{x+c}\right)^{d}\right) \right) \right) \\
    &= \left( 1 + \mathcal{O}\left(\left(\frac{x}{x+c}\right)^{d}\right) \right)^{\min(n,m)} \underset{1\leqslant i,j \leqslant \min(n,m)}{\mathrm{det}} \left( \binom{d+i+j-2}{i-1} \right) \\
    &=  \left( 1 + \mathcal{O}\left(\left(\frac{x}{x+c}\right)^{d} \right) \right) \times 1.
\end{align}
using the fact that $\min(n,m)$ is a constant, $0 < \frac{x}{x+c} < 1$ and the proposition \ref{propdet}. Therefore 
\begin{equation}
     Z_{nm}(x|B|\hat C) = Z_{nm}^0(B|\hat C) \left(1+\frac{x}{c}\right)^{nm} \left( 1 + \mathcal{O}\left(\frac{x}{x+c}\right)^{d} \right).
\end{equation}
This expression allows us to obtain an asymptotically equivalent (once again because $0 < \frac{x}{x+c} < 1$) :
\begin{equation}
    Z_{nm}(x|B|\hat C) \underset{\max(n,m) \rightarrow +\infty}{\sim} Z_{nm}^0(B|\hat C) \left(1+\frac{x}{c}\right)^{nm}.
\end{equation}

\begin{remark}
    Asymptotically, it is like all vertices were of type $a$. This Boltzmann weight corresponds to the lower energy level (because $a(x) = b(x) + c(x)$ and $b(x), c(x) \geqslant 0$). So, it is like the lattice was at the fundamental energy. We see also the boundary effect.
\end{remark}

\subsubsection{The infinite lattice}\label{infinitelattice}

We consider the limit $m,n \rightarrow +\infty$ with $d$ constant. The first part of the calculation is based on \cite{Korepin2000}. 
We consider the expression of $Z_{nm}(x|B|\hat C)$ \eqref{Zhomo1}-\eqref{Zhomo2}. We denote 
\begin{equation}
    \tau_{n,m}(x) = \underset{1 \leqslant i,j \leqslant \min(n,m)}{\mathrm{det}} \left( \phi_\beta^{(d+i+j-2)}(x) \right).
\end{equation}

Because $d = cste$, $\tau_{n,m}$ is a Hankël determinant and satisfies the Toda equation (in Hirota form) \cite{Sogo1993} 
\begin{equation}\label{Todaeq}
    \tau_{n+1,m+1}\tau_{n-1,m-1} = \tau_{n,m}\tau''_{n,m} - (\tau'_{n,m})^2
\end{equation}
where the prime denotes the $x$ derivative. 
When $n,m \rightarrow +\infty$ it is expected that 
\begin{equation}\label{ZetF}
    \ln(Z_{nm}(x|B|\hat C)) = -nm\frac{F(x)}{k_BT} + o(nm)
\end{equation}
where $F(x)$ is the bulk free energy. We denote also
\begin{equation}
    f(x) = -\frac{F(x)}{k_BT} + \ln\left(\frac{\phi_1(x)}{c}\right).
\end{equation}
Expressing $\tau_{n,m}$ in terms of $f(x)$ and constant factors, in the same way as in \cite{Korepin2000}, the Toda equation becomes 
\begin{equation}
    f'' = \exp(2f)
\end{equation}
and the solutions of this equation have the form
\begin{equation}
 \exp(f(x)) = \frac{\alpha}{\sinh(\alpha (x-x_0))}.
\end{equation}
So we can write the bulk free energy as
\begin{equation}
 -\frac{F(x)}{k_BT} = \ln\left(e^{f(x)}\frac{c}{\phi_1(x)}\right) = \ln\left(\frac{\alpha}{\sinh(\alpha(x-x_0))}\frac{x(x+c)}{c}\right).
\end{equation}

Now we go to determine $\alpha$ and $x_0$ with knowns values of $F$. First, we know that $\ln(Z_{nm}(0\vert B \vert \hat C)) = o(nm)$ because the constants in $Z^0_{nm}(B|\hat C)$ are at the power $n$ or $m$. So $F(0) = 0$. However, for $x_0 = 0$,
\begin{equation}
 \frac{\alpha}{\sinh(\alpha x)}\frac{x(x+c)}{c} \xrightarrow[x \rightarrow 0]{} 1
\end{equation}
and we obtain $F(0) = 0$. So we can take $x_0 = 0$. 
Now, for $\alpha$, the Taylor expansion at order 2 of $x \mapsto \ln\left(\frac{\alpha}{\sinh(\alpha x)}\frac{x(x+c)}{c}\right)$ gives us
\begin{equation}
 \ln\left(\frac{\alpha}{\sinh(\alpha x)}\frac{x(x+c)}{c}\right) = \frac{x}{c}-\frac12\left(\frac{\alpha^2}{3}+\frac{1}{c^2}\right)x^2 + o\left(x^2\right)
\end{equation}
so 
\begin{equation}
 F(0) = 0, \qquad \frac{F'(0)}{k_BT} = -\frac{1}{c} \qquad \text{and} \qquad \frac{F''(0)}{k_BT} = \frac{\alpha^2}{3}+\frac{1}{c^2}.
\end{equation}
On the other hand, deriving \eqref{ZetF} once and twice and evaluating at $0$, we obtain
\begin{equation}\label{ZFp}
 \frac{Z_{nm}'(0\vert B \vert \hat C)}{Z^0_{nm}(B \vert \hat C)} = -nm\frac{F'(0)}{k_BT} + o(nm),
\end{equation}
\begin{equation}\label{ZFs}
 \frac{Z_{nm}''(0\vert B \vert \hat C)}{Z^0_{nm}(B \vert \hat C)} - \left(\frac{Z_{nm}'(0\vert B \vert \hat C)}{Z^0_{nm}(B \vert \hat C)}\right)^2 = -nm\frac{F''(0)}{k_BT} + o(nm).
\end{equation}
Computing $Z_{nm}'(0\vert B \vert \hat C)$ and $Z_{nm}''(0\vert B \vert \hat C)$ and injecting their expression in the equations \eqref{ZFp}-\eqref{ZFs} (see the subsection \ref{derivativeZ} for details), we observe that the value of $\alpha$ depend on the one of $d$ : 
\begin{itemize}
 \item if $d > 1$, $\alpha = 0$,
 \item if $d = 1$, $\alpha^2 = \frac{3\beta}{c^2}$,
 \item if $d = 0$ (\textit{i.e.} $n = m$), $\alpha^2 = \frac{3\beta(\beta-2)}{c^2}$.
\end{itemize}
So the bulk free energy depends also on the value of $\vert n-m \vert$.
For $d > 1$, there is no boundary effect and it is like all the vertices are of type $a$ :
 \begin{equation}
  \frac{F(x)}{k_BT} = \ln\left(\frac{c}{x+c} \right).
 \end{equation}
 For $d \leqslant 1$, 
\begin{equation}
 \frac{F(x)}{k_BT} = \ln\left(\frac{\sinh(\alpha x)}{\alpha x}\frac{c}{x+c} \right) = \ln\left(\frac{\sinh(\alpha x)}{\alpha x} \right) + \ln\left(\frac{c}{x+c} \right).
\end{equation}
There is an additional term in the expression of the bulk free energy. This term comes from the boundary effects because it depends on $\beta$.

We want to simplify this expression. For simplicity, we define
\begin{equation}
    \tilde \beta = \left\{ \begin{array}{cc}
         0 &  \text{if} \;\; d > 1\\
         3\beta & \text{if} \;\; d = 1\\
         3\beta(\beta-2) & \text{if} \;\; d = 0
    \end{array} \right. 
\end{equation}
so that $\alpha^2 = \frac{\tilde \beta}{c^2}$. Moreover, we recall that $x,c > 0$ and the free energy must be real. So $\alpha$ must be real or imaginary pure and $\tilde \beta$ must be real.
In addition, the function $\alpha \mapsto \frac{\sinh(\alpha x)}{\alpha x}$ is even so we can consider only one square root of $\tilde \beta$. According to the sign of $\tilde \beta$, the free energy has different forms. Using the variable $\tilde x = \frac{x}{c}$, we obtain
\begin{equation}
    \tilde F\left(\tilde x \right) = \frac{F\left(c \tilde x \right)}{k_BT} = \left\{
    \begin{array}{cc}
        \ln\left(\frac{\sinh\left(\sqrt{\tilde\beta} \tilde x \right)}{\sqrt{\tilde\beta} \tilde x} \frac{1}{1+\tilde x} \right) & \text{if} \;\; \tilde \beta > 0 \\
        \ln\left(\frac{1}{1+\tilde x} \right) & \text{if} \;\; \tilde \beta = 0 \\ 
        \ln\left(\frac{\sin\left(\sqrt{-\tilde\beta} \tilde x \right)}{\sqrt{-\tilde\beta} \tilde x} \frac{1}{1+\tilde x} \right) & \text{if} \;\; \tilde \beta < 0 
    \end{array} \right. .
\end{equation}

We represent in the figure \ref{fig:F} the graph of $\tilde F$ for different values of $\tilde \beta$. In the figure \ref{fig:Fmid}, we have represented the three different possible forms for $\tilde F$. 


Apart from when $\tilde x \rightarrow 0$ ($\tilde F(0) = -1$), the three functions have different graphs. So the case $\tilde \beta =0$ might be correspond to a phase transition. 
For $\tilde \beta > 0$, as we see in the figure \ref{fig:Fpos}, we observe that $\tilde F(\tilde x) \xrightarrow[\tilde x \rightarrow +\infty]{} + \infty$. Moreover, we can show that $\tilde F$ has a parabolic branch of equation $y = \tilde \beta \tilde x$. Therefore, the larger $\tilde \beta$ is, the faster $\tilde F$ tends towards infinity.
For $\tilde \beta \leqslant 0$, as we see in the figure \ref{fig:Fneg}, $\tilde F(\tilde x) \xrightarrow[\tilde x \rightarrow +\infty]{} -\infty$.
Moreover, for $\tilde \beta < 0$, $\tilde F$ is well defined only on the intervals $\left[0, \frac{\pi}{\sqrt{-\tilde \beta}}\right[$ and $\left] \frac{2k\pi}{\sqrt{-\tilde \beta}},  \frac{(2k+1)\pi}{\sqrt{-\tilde \beta}} \right[$, with $k \in \NN^*$.

We can also compute the physical characteristics of the system. We recall that $\tilde x$ is a Boltzmann weight. We have $F = k_B T \tilde F \left( \tilde x \right)$ so
\begin{align}
    \frac{\partial F}{\partial T} &= k_B \tilde F(\tilde x) + k_B T \frac{\deriv \tilde x}{\deriv T} \frac{\deriv \tilde F}{\deriv \tilde x} \\
    &= k_B \tilde F(\tilde x) + k_B T \left( -\frac{\tilde x \ln\tilde x }{T} \right)\frac{\deriv \tilde F}{\deriv \tilde x} = k_B \tilde F(\tilde x) - k_B \tilde x \ln\tilde x \frac{\deriv \tilde F}{\deriv \tilde x}
\end{align}
and 
\begin{align}
    \frac{\partial^2 F}{\partial T^2} &= k_B \frac{\deriv \tilde x}{\deriv T} \frac{\deriv \tilde F}{\deriv \tilde x} - k_B (\ln\tilde x -1) \frac{\deriv \tilde x}{\deriv T}\frac{\deriv \tilde F}{\deriv \tilde x} - k_B \tilde x \ln\tilde x \frac{\deriv \tilde x}{\deriv T}\frac{\deriv^2 \tilde F}{\deriv \tilde x^2}  \\
    &= \frac{k_B}{T} \tilde x \ln\tilde x \left( \ln \tilde x - 2 \right) \frac{\deriv \tilde F}{\deriv \tilde x} + \frac{k_B}{T} \tilde x^2 \left(\ln\tilde x\right)^2 \frac{\deriv^2 \tilde F}{\deriv \tilde x^2}.
\end{align}
So the average energy, the fluctuation of the energy, the heat capacity at constant volume and the canonical entropy are respectively
\begin{equation}
    \bar{E}^c_\infty = F - T \frac{\partial F}{\partial T} = k_BT\tilde x \ln\tilde x \frac{\deriv \tilde F}{\deriv \tilde x},
\end{equation}
\begin{equation}
    \left( \Delta E^c_\infty \right)^2 = -k_BT^3  \frac{\partial^2 F}{\partial T^2} = (k_BT)^2 \tilde x \ln\tilde x \left( \left( 2-\ln \tilde x \right) \frac{\deriv \tilde F}{\deriv \tilde x} - \tilde x \ln\tilde x \frac{\deriv^2 \tilde F}{\deriv \tilde x^2} \right),
\end{equation}
\begin{equation}
    S^c_\infty = \frac{C_{V,\infty}}{T} =  - \frac{\partial F}{\partial T} =  - k_B  \tilde F \left(\tilde x\right) + k_B\tilde x \ln \tilde x \frac{\deriv \tilde F}{\deriv \tilde x}.
\end{equation}
We represent these different quantities, with $k_BT = 1$, for $\beta \geqslant 0$ in the figure \ref{fig:charac_beta_pos} and for $\tilde \beta = -1$ in the figures \ref{fig:charac_beta_neg} and \ref{fig:charac_beta_neg_near_0}. 

For $\tilde \beta > 0$, $\bar{E}^c_\infty$ has a negative part between two positive parts, with a minimum in the negative part, and tends towards infinity when $\tilde x \rightarrow + \infty$ (see figure \ref{fig:Epos}). $S^c_\infty$ has also a minimum value (positive then negative when $\tilde \beta$ increases) and tends towards infinity when $\tilde x \rightarrow + \infty$ (see figure \ref{fig:Spos}). And $(\Delta E^c_\infty)^2$ tends towards minus infinity when $\tilde x \rightarrow \infty$ and has one maximum or one minimum between two maxima, when $\tilde \beta$ increases (see figure \ref{fig:VEpos}).

\begin{figure}[h]
    \centering
    \begin{subfigure}[h]{0.5\textwidth}
        \centering
        \includegraphics[width=\textwidth]{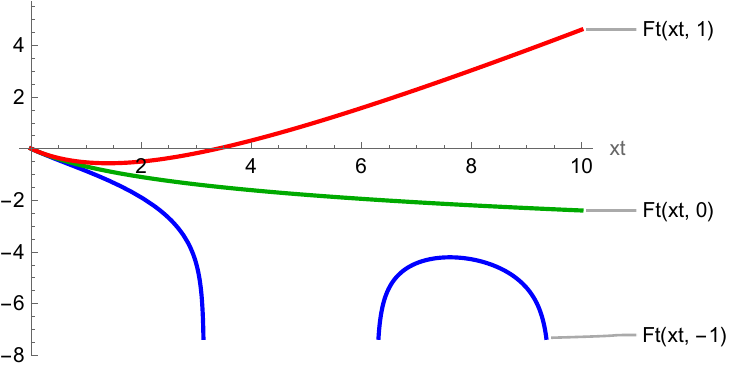}
        \caption{\centering For $\tilde \beta = -1, 0$ and $1$ (respectively in blue, green and red).}
        \label{fig:Fmid}
    \end{subfigure}
    \begin{subfigure}[h]{0.48\textwidth}
        \centering
        \includegraphics[width=\textwidth]{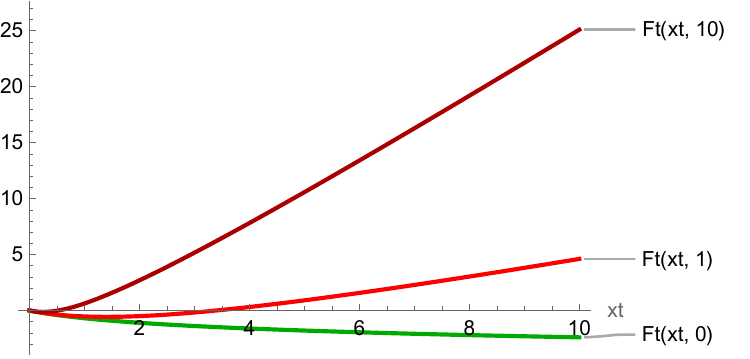}
        \caption{\centering For $\tilde \beta = 0, 1$ and $10$ (respectively in green, red and dark-red).}
        \label{fig:Fpos}
    \end{subfigure}
    \hspace{1mm}
    \begin{subfigure}[h]{0.48\textwidth}
        \centering
        \includegraphics[width=\textwidth]{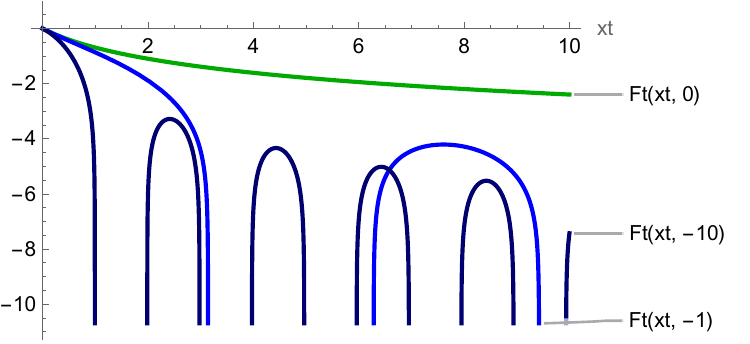}
        \caption{\centering For $\tilde \beta = 0, -1$ and $-10$ (respectively in green, blue and dark-blue).}
        \label{fig:Fneg}
    \end{subfigure}   
    \caption{\centering Graphs of $\tilde F$ for different values of $\tilde \beta$. \enquote{xt} corresponds to $\tilde x$ and \enquote{Ft} to $\tilde F$.}
    \label{fig:F}
\end{figure}

For $\tilde \beta = 0$, we can explicit the formulae of the physical characteristics
\begin{equation}
    \bar{E}^c_\infty = -\frac{\tilde x \ln \tilde x}{1+\tilde x}, \;\;\; 
\end{equation} 
\begin{equation}
    \left(\Delta E^c_\infty\right)^2 = - \frac{\tilde x\ln\tilde x}{(1+\tilde x)^2}(2+2\tilde x - \ln \tilde x),
\end{equation}
\begin{equation}
    S^c_\infty = \frac{C_{V,\infty}}{T} = -\frac{\tilde x \ln \tilde x}{1+\tilde x} + \ln\left(1+\tilde x\right).
\end{equation}
In addition, with the graphs on the figure \ref{fig:charac_beta_pos}, we see that $\bar{E}^c_\infty$ and $\left(\Delta E^c_\infty\right)^2$ have one maximum, one positive part and one negative part and tend toward minus infinity when $\tilde x \rightarrow \infty$. In contrast, $S^c_\infty$ has one maximum, it is only positive and tends toward 0 when $\tilde x \rightarrow \infty$. 

For $\tilde \beta < 0$, we have represented in the figures \ref{fig:charac_beta_neg} and \ref{fig:charac_beta_neg_near_0} the physical characteristics only on the domain of the free energy. In the figure \ref{fig:Eneg} and \ref{fig:Sneg}, we see that $\bar{E}^c_\infty$ and $S^c_\infty$ fluctuate between positive and negative values and tend toward plus or minus infinity when $\tilde x \rightarrow \frac{k\pi}{\sqrt{-\tilde \beta}}$ for $k \in \NN^*$, whereas $\left(\Delta E^c_\infty\right)^2$ seems only positive, has local minima and tends toward plus infinity when $\tilde x \rightarrow \frac{k\pi}{\sqrt{-\tilde \beta}}$ (see the figure \ref{fig:VEneg}).
Near to $\tilde x = 0$, $\bar{E}^c_\infty$ and $S^c_\infty$ have a positive and a negative part and a maximum value (see the figures \ref{fig:Eneg_near_0} and \ref{fig:Sneg_near_0}). Finally, as we see in the figure \ref{fig:VEneg_near_0}, $\left(\Delta E^c_\infty\right)^2$ has a negative part between two positive ones and a maximum and a minimum value .

\vspace{5mm}

\begin{figure}[h]
    \centering
    \begin{subfigure}[h]{0.45\textwidth}
        \centering
        \includegraphics[width=\textwidth]{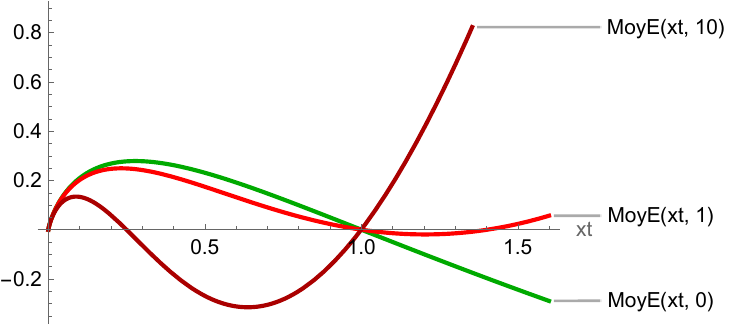}
        \caption{\centering Graph of $\bar{E}^c_\infty$.}
        \label{fig:Epos}
    \end{subfigure}
    \begin{subfigure}[h]{0.45\textwidth}
        \centering
        \includegraphics[width=\textwidth]{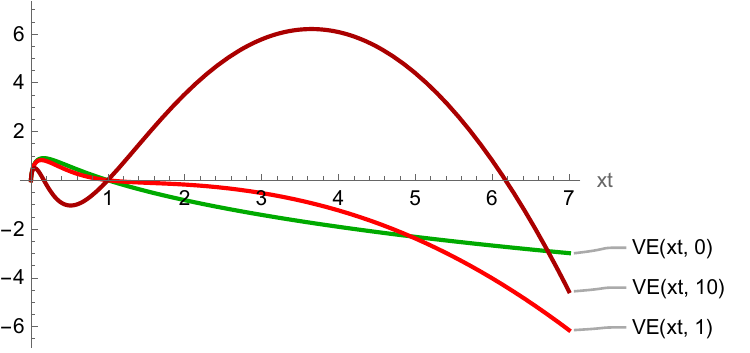}
        \caption{Graph of $(\Delta E^c_\infty)^2$.}
        \label{fig:VEpos}
    \end{subfigure}
    \hspace{1mm}
    \begin{subfigure}[h]{0.45\textwidth}
        \centering
        \includegraphics[width=\textwidth]{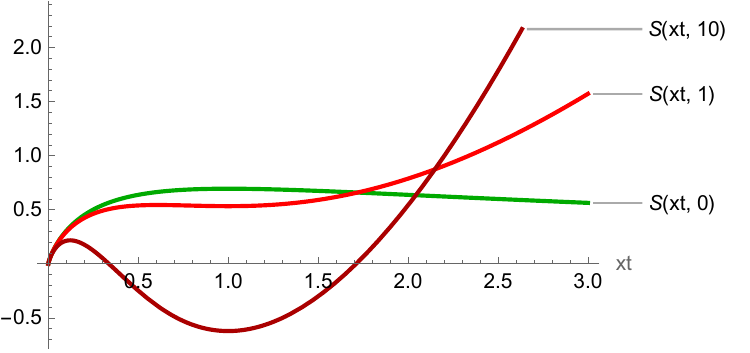}
        \caption{\centering Graph of $S^c_\infty = \frac{C_{V,\infty}}{T}$.}
        \label{fig:Spos}
    \end{subfigure}   
    \caption{\centering Graphs of the physical characteristics for $k_BT = 1$ and $\tilde \beta \geqslant 0$. The green, red and dark-red curve are respectively for $\tilde \beta = 0, 1$ and $10$.}
    \label{fig:charac_beta_pos}
\end{figure}

\vspace{5mm}

\begin{figure}[h]
    \centering
    \begin{subfigure}[h]{0.45\textwidth}
        \centering
        \includegraphics[width=\textwidth]{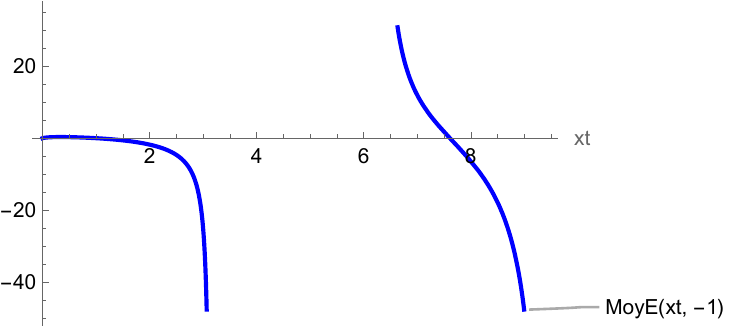}
        \caption{\centering Graph of $\bar{E}^c_\infty$.}
        \label{fig:Eneg}
    \end{subfigure}
    \begin{subfigure}[h]{0.45\textwidth}
        \centering
        \includegraphics[width=\textwidth]{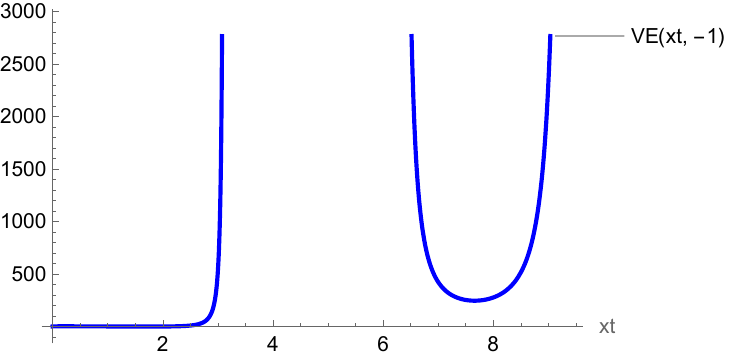}
        \caption{Graph of $(\Delta E^c_\infty)^2$.}
        \label{fig:VEneg}
    \end{subfigure}
    \hspace{1mm}
    \begin{subfigure}[h]{0.45\textwidth}
        \centering
        \includegraphics[width=\textwidth]{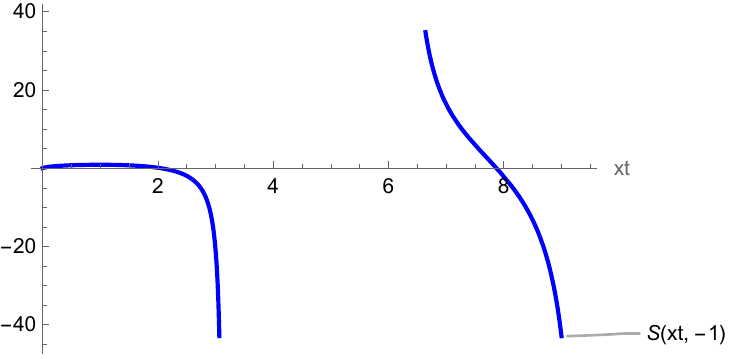}
        \caption{\centering Graph of $S^c_\infty = \frac{C_{V,\infty}}{T}$.}
        \label{fig:Sneg}
    \end{subfigure}   
    \caption{\centering Graphs of the physical characteristics for $k_BT = 1$ and $\tilde \beta =-1$.}
    \label{fig:charac_beta_neg}
\end{figure}

\newpage

\begin{figure}[h]
    \centering
    \begin{subfigure}[h]{0.45\textwidth}
        \centering
        \includegraphics[width=\textwidth]{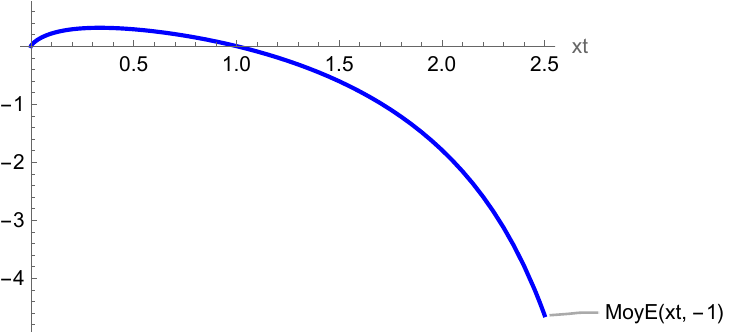}
        \caption{\centering Graph of $\bar{E}^c_\infty$.}
        \label{fig:Eneg_near_0}
    \end{subfigure}
    \begin{subfigure}[h]{0.45\textwidth}
        \centering
        \includegraphics[width=\textwidth]{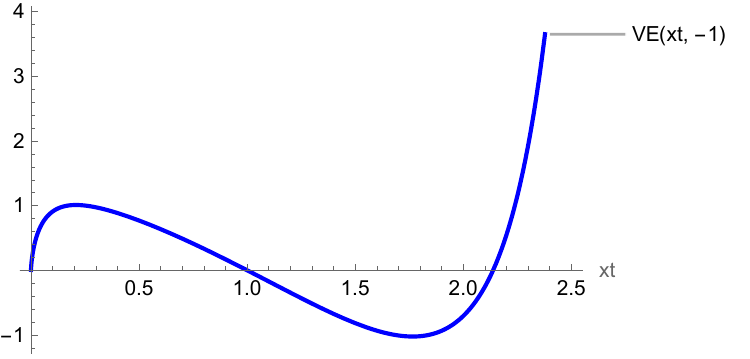}
        \caption{Graph of $(\Delta E^c_\infty)^2$.}
        \label{fig:VEneg_near_0}
    \end{subfigure}
    \hspace{1mm}
    \begin{subfigure}[h]{0.45\textwidth}
        \centering
        \includegraphics[width=\textwidth]{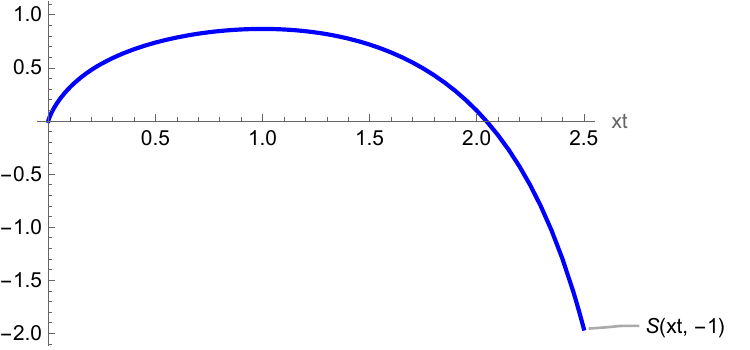}
        \caption{\centering Graph of $S^c_\infty = \frac{C_{V,\infty}}{T}$.}
        \label{fig:Sneg_near_0}
    \end{subfigure}   
    \caption{\centering Graphs of the physical characteristics for $k_BT = 1$ and $\tilde \beta =-1$ (for $\tilde x \leqslant 2.5$).}
    \label{fig:charac_beta_neg_near_0}
\end{figure}

\newpage

\section*{Conclusion}

\addcontentsline{toc}{section}{Conclusion}

Thanks to the method developed by Foda and Wheeler \cite{Foda2012} for the pDWBC, based on computations using the determinant properties, we have found a new expression \eqref{newformulaZ} for the partition function of the MR6V model, introduced in \cite{Belliard2024}. 

This new expression has enabled us to compute the homogeneous limit of this model, using the historical technique with the Taylor expansion, introduced in \cite{Izergin1992}. 
Then, we have computed the thermodynamic limit in two cases. In the case of the infinite lattice (limits in both directions), using some computations of \cite{Korepin2000}, we have found that the free energy of the system and the linked physical characteristics depend on the parameter $\beta$ \eqref{defbeta},  in other words, on the boundary conditions. However, our asymptotic approximation \eqref{ZetF} may miss important effects for strong boundary fields and other orders should be considered.

These results have to be generalized to the trigonometric (XXZ) six vertex model and further to the eight vertex model, related to the XYZ spin chain. The key step is to consider the MABA in these cases to find the associated Bethe vectors, that allows us to obtain an algebraic expression of the partition function for general boundary conditions.
The existence of the elliptic and trigonometric source identities \cite{Motegi2024a}, in the theory of multivariable special functions, strengthens our belief that there exists formulae of these partition functions in terms of Bethe vectors for trigonometric and elliptic cases. Indeed, the authors make a link between the Modified Izergin determinant in the rational case and the rational source identity. Moreover, a formula already exists for the elliptic SOS model with DWBC \cite{Galleas2018,Pakuliak2008,Rosengren2008} (this model is very close to the eight vertex ones). 

Moreover, we suppose that the physical results in the thermodynamic limit will be more interesting. In itself, the results might depend on the boundary conditions. 
Furthermore, for DWBC, the free energy in the thermodynamic limit depends on the phase of the system (see \cite{Korepin2000,ZinnJustin2000}), phases depending on the anisotropy of the system. In particular, for the anti-ferroelectric phase and at zero temperature limit, it appears an arctic curve that separate the central disordered phase of the four ordered corners. This arctic curve appears also for other boundary conditions (for pDWBC, reflective end ... see for example \cite{Lyberg2018} for numerical results and graphics).

Finally, we can also consider the linear system approach to the computation of the partition function of the six vertex model with reflecting boundary conditions (in the rational, trigonometric and elliptic case), related to the reflection algebra \cite{Sklyanin1988,Tsuchiya1998}.



\section*{Acknowledgments}

The authors would like to thank Nikita A. Slavnov for interesting discussions and the referees for reviewing, reporting and interesting comment on the article. M. Cornillault would like to thank also the ENS of Rennes for funding his thesis.


\newpage

\appendix

\section{Proofs of the new formula for the partition function}\label{appendixA}

\subsection{Using the method of Foda and Wheeler}

This proof follows the same reasoniong as in \cite{Foda2012}. We begin to prove the proposition for the case $n \leqslant m$. We prove the following statement by induction over $k \in [\![0,m-n]\!]$
 \begin{align}
    Z_{m-k,m}(\bu|\bv|B|\hat C) = \frac{\mathrm{tr}(B)^{m-k}\mathrm{tr}(\hat C)^m}{(1-\beta)^{m-k}} &\frac{h(\bu,\bv)}{g(\bu,\bv)\Delta(\bu)\Delta'(\bv)} \notag \\ 
    & \times \underset{1 \leqslant i,j \leqslant m}{\mathrm{det}} \left( \left\{ \begin{array}{cc}
        \phi_\beta(u_i-v_j) & \text{if} \;\; i \leqslant m-k \\
        \psi_{m-i}(-v_j) & \text{if} \;\; i \geqslant m-k+1
    \end{array} \right. \right).
\end{align}

For $k = 0$, there is nothing to prove because the explicit formula is already known with this form (see equations \eqref{ZwithK}-\eqref{K3}).
We detail now the case $k = 1$ to see the construction of the determinant. By definition, 
\begin{equation}
  Z_{nm}(\bu|\bv|B|\hat C) = \lim\limits_{u_m \rightarrow \infty} \left( \frac{1}{\mathrm{tr}(B)} \left(\frac{c}{u_m}\right)^m Z_{mm}(\bu\vert\bv|B|\hat{C}) \right)
\end{equation}
and 
\begin{equation}
 Z_{mm}(\bu|\bv|B|\hat C) = \frac{\mathrm{tr}(B)^m\mathrm{tr}(\hat C)^m}{(1-\beta)^m} \frac{h(\bu,\bv)}{g(\bu,\bv)\Delta(\bu)\Delta'(\bv)} \underset{1 \leqslant i,j \leqslant m}{\mathrm{det}} \left( \phi_\beta(u_i-v_j) \right).
\end{equation}
So 
\begin{align}
 \frac{1}{\mathrm{tr}(B)} &\left(\frac{c}{u_m}\right)^m Z_{mm}(\bu\vert\bv|B|\hat{C}) = \frac{\mathrm{tr}(B)^{m-1}\mathrm{tr}(\hat C)^m}{(1-\beta)^m} \frac{h(\bu_m,\bv)}{g(\bu_m,\bv)\Delta(\bu_m)\Delta'(\bv)} \notag \\ 
 & \times \left[ \left( \frac{c}{u_m} \right)^m \prod_{j=1}^m \frac{(u_m-v_j+c)(u_m-v_j)}{c^2} \prod_{i=1}^{m-1} \frac{c}{u_m-u_i}\right] \underset{1 \leqslant i,j \leqslant m}{\mathrm{det}} \left( \phi_\beta(u_i-v_j) \right).
\end{align}
We make the change of variable $u_m = \frac{1}{\varepsilon}$ and we inject the factor between brackets in the last row of the determinant (the row $m$). The coefficients $d_{m,\tilde \jmath}$, for $\tilde \jmath \in [\![1,m]\!]$, become
\begin{align}
 d_{m,\tilde \jmath} &= (\varepsilon c)^m \left( \frac{\varepsilon c}{1-v_{\tilde \jmath} \varepsilon} - \frac{\beta \varepsilon c}{1-(v_{\tilde \jmath}-c)\varepsilon} \right) \prod_{j=1}^m \frac{(1-(v_j-c)\varepsilon)(1-v_j\varepsilon)}{(\varepsilon c)^2} \prod_{i=1}^{m-1} \frac{\varepsilon c}{1-u_i \varepsilon}  \\
 &= \left( \frac{1}{1-v_{\tilde \jmath} \varepsilon} - \frac{\beta}{1-(v_{\tilde \jmath}-c)\varepsilon} \right) \prod_{j=1}^m (1-(v_j-c)\varepsilon)(1-v_j\varepsilon) \prod_{i=1}^{m-1} \frac{1}{1-u_i \varepsilon} 
\end{align}
So $d_{m,\tilde \jmath} \xrightarrow[\varepsilon \rightarrow 0]{} 1 - \beta$. Factorizing out of the determinant by $1-\beta$, we prove the statement for $k=1$.

We suppose now that the statement is true for one $k \in [\![0,m-n-1]\!]$ and we consider the statement for $k+1$. By definition, 
\begin{align}
 Z_{m-(k+1),m}(\bu|\bv|B|\hat C) &= \lim\limits_{u_m, \dots , u_{m-k} \rightarrow \infty} \left( \frac{1}{\mathrm{tr}(B)^{m-n}} \prod_{i=n+1}^m \left(\frac{c}{u_i}\right)^m Z_{mm}(\bu\vert\bv|B|\hat{C}) \right) \\
 &= \lim\limits_{u_{m-k} \rightarrow \infty} \left( \frac{1}{\mathrm{tr}(B)} \left(\frac{c}{u_{m-k}}\right)^m Z_{m-k,m}(\bu|\bv|B|\hat C) \right).
\end{align}
According the induction statement, 
\begin{align}
    Z_{m-k,m}(\bu|\bv|B|\hat C) = \frac{\mathrm{tr}(B)^{m-k}\mathrm{tr}(\hat C)^m}{(1-\beta)^{m-k}} &\frac{h(\bu,\bv)}{g(\bu,\bv)\Delta(\bu)\Delta'(\bv)} \notag \\ 
    & \times \underset{1 \leqslant i,j \leqslant m}{\mathrm{det}} \left( \left\{ \begin{array}{cc}
        \phi_\beta(u_i-v_j) & \text{if} \;\; i \leqslant m-k \\
        \psi_{m-i}(-v_j) & \text{if} \;\; i \geqslant m-k+1
    \end{array} \right. \right).
\end{align}
So
\begin{align}
 \frac{1}{\mathrm{tr}(B)} \left(\frac{c}{u_{m-k}}\right)^m &Z_{m-k,m}(\bu|\bv|B|\hat C) = \frac{\mathrm{tr}(B)^{m-(k+1)}\mathrm{tr}(\hat C)^m}{(1-\beta)^{m-k}} \frac{h(\bu_{m-k},\bv)}{g(\bu_{m-k},\bv)\Delta(\bu_{m-k})\Delta'(\bv)} \notag \\
 & \times \left[ \left( \frac{c}{u_{m-k}} \right)^m \prod_{j=1}^m \frac{(u_{m-k}-v_j+c)(u_{m-k}-v_j)}{c^2} \prod_{i=1}^{m-k-1} \frac{c}{u_{m-k}-u_i} \right] \notag \\
 &\times \underset{1 \leqslant i,j \leqslant m}{\mathrm{det}} \left( \left\{ \begin{array}{cc}
        \phi_\beta(u_i-v_j) & \text{if} \;\; i \leqslant m-k \\
        \psi_{m-i}(-v_j) & \text{if} \;\; i \geqslant m-k+1
    \end{array} \right. \right).
\end{align}
We make the change of variable $u_{m-k} = \frac{1}{\varepsilon}$ and we inject the factor between brackets in the row $m-k$ of the determinant. The coefficients $d_{m-k,\tilde \jmath}$, for $\tilde \jmath \in [\![1,m]\!]$, become 
\begin{align}
 d_{m-k,\tilde \jmath} &= (\varepsilon c)^m \left( \frac{\varepsilon c}{1-v_{\tilde \jmath}\varepsilon} - \frac{\beta \varepsilon c}{1-(v_{\tilde \jmath}-c)\varepsilon} \right) \prod_{j=1}^m \frac{(1-v_j\varepsilon)(1-(v_j-c)\varepsilon)}{(\varepsilon c)^2} \prod_{i=1}^{m-k-1} \frac{\varepsilon c}{1-u_i\varepsilon}  \\
 &= \frac{1}{(\varepsilon c)^k} \left( \frac{1}{1-v_{\tilde \jmath}\varepsilon} - \frac{\beta}{1-(v_{\tilde \jmath}-c)\varepsilon} \right) \prod_{j=1}^m (1-v_j\varepsilon)(1-(v_j-c)\varepsilon) \prod_{i=1}^{m-k-1} \frac{1}{1-u_i\varepsilon} .
\end{align}
We use the elementary $e_l$ and complete symmetric functions $h_l$ defined in the subsection \ref{symmfunc} with their generating function \eqref{genersymmfunc}.
Therefore
\begin{equation}
 \prod_{j=1}^m (1-v_j\varepsilon)(1-(v_j-c)\varepsilon) = \sum_{l=0}^{\infty} (-1)^l e_l\left(\frac{\bv}{c} \bigcup \left( \frac{\bv}{c}-1 \right) \right)(\varepsilon c)^l,
\end{equation}
\begin{equation}
 \prod_{i=1}^{m-k-1} (1-u_i\varepsilon)^{-1} = \sum_{l=0}^{\infty} h_l\left(\frac{\bu_{m-k}}{c}\right)(\varepsilon c)^l,
\end{equation}
\begin{equation}
 \frac{1}{1-v_{\tilde \jmath}\varepsilon}-\frac{\beta}{1-(v_{\tilde \jmath}-c)\varepsilon} = \sum_{l=0}^{\infty} \left( \left(\frac{v_{\tilde \jmath}}{c}\right)^l-\beta \left( \frac{v_{\tilde \jmath}}{c}-1\right)^l \right)(\varepsilon c)^l.
\end{equation}
So
\begin{align}
 d_{m-k,\tilde \jmath} = \sum_{l=0}^{\infty} \left(\sum_{l_1+l_2+l_3=l} f_{l_1,l_2}(\bu_{m-k}\vert\bv) \left( \left(\frac{v_{\tilde \jmath}}{c}\right)^{l_3}-\beta \left( \frac{v_{\tilde \jmath}}{c}-1\right)^{l_3} \right) \right) (\varepsilon c)^{l-k}
\end{align}
with 
\begin{equation}
 f_{l_1,l_2}(\bu_{m-k}\vert\bv) = (-1)^{l_1} e_{l_1}\left(\frac{\bv}{c} \bigcup \left( \frac{\bv}{c}-1 \right) \right) h_{l_2}\left(\frac{\bu_{m-k}}{c}\right).
\end{equation}
However,
\begin{align}
 \left(\frac{v_{\tilde \jmath}}{c}\right)^{l_3} - \beta \left( \frac{v_{\tilde \jmath}}{c}-1\right)^{l_3} = (1-\beta)\left(\frac{v_{\tilde \jmath}}{c}\right)^{l_3} - \beta \sum_{l_4=0}^{l_3-1} \binom{l_3}{l_4} \left(\frac{v_{\tilde \jmath}}{c}\right)^{l_4} (-1)^{l_3-l_4}.
\end{align}
So the terms for $l < k$ in $d_{m-k,\tilde\jmath}$ do not contribute to the determinant because there are a linear combination of the rows below. For $l=k$, only the term $(l_1=0,l_2=0,l_3=k)$ contribute to the determinant for the same reasons. And for $l > k$, the terms tend to $0$ when $\varepsilon \rightarrow 0$. We can therefore write the row $m-k$ with the term $(1-\beta)\left(\frac{v_{\tilde \jmath}}{c}\right)^k$ for all $\tilde \jmath \in [\![1,m]\!]$ when $\varepsilon \rightarrow 0$. Factorizing $1-\beta$ out of the determinant, we obtain 
\begin{align}
    Z_{m-(k+1),m}(\bu|\bv|B|\hat C) = &\frac{\mathrm{tr}(B)^{m-(k+1)}\mathrm{tr}(\hat C)^m}{(1-\beta)^{m-(k+1)}} \frac{h(\bu,\bv)}{g(\bu,\bv)\Delta(\bu)\Delta'(\bv)} \notag \\ 
    & \times \underset{1 \leqslant i,j \leqslant m}{\mathrm{det}} \left( \left\{ \begin{array}{cc}
        \phi_\beta(u_i-v_j) & \text{if} \;\; i \leqslant m-(k+1) \\
        \psi_{m-i}(-v_j) & \text{if} \;\; i \geqslant m-k
    \end{array} \right. \right)
\end{align}
which proves the statement for $k+1$ and finishes the induction proof.

Applying the formula for $k = m-n$, we obtain
\begin{align}
    Z_{n,m}(\bu|\bv|B|\hat C) = \frac{\mathrm{tr}(B)^{n}\mathrm{tr}(\hat C)^m}{(1-\beta)^{n}} &\frac{h(\bu,\bv)}{g(\bu,\bv)\Delta(\bu)\Delta'(\bv)} \notag \\ 
    & \times \underset{1 \leqslant i,j \leqslant m}{\mathrm{det}} \left( \left\{ \begin{array}{cc}
        \phi_\beta(u_i-v_j) & \text{if} \;\; i \leqslant n \\
        \psi_{m-i}(-v_j) & \text{if} \;\; i \geqslant n+1
    \end{array} \right. \right)
\end{align}
which concludes the proof for $n \leqslant m$.

To prove the formula for $n \geqslant m$, we just have to exchange in the previous proof the following things : $n \leftrightarrow m$, $B \leftrightarrow \hat C$, $u_i \leftrightarrow -v_j$ and $i \leftrightarrow j$. Then, the proof follows the same reasoning.

\subsection{An alternative proof}\label{equalitynewknown}

Now we will prove that each formula of \eqref{newformulaZ} is equal to one of the two known formulae (\eqref{ZwithK} with \eqref{K1} or \eqref{K2}). Contrary to the previous proof, this method is original.

To do this, as in \cite{Belliard2024}, we consider that $g(\bu,\bv)\Delta(\bu)\Delta'(\bv)$ is a determinant of some \enquote{Cauchy} matrix and we multiply the matrix in the determinant by the inverse of the \enquote{Cauchy} matrix. Then, we will transform the obtained matrix using the properties of the determinant. 
We denote the Cauchy matrix in quotation marks because, for the real Cauchy matrix, we have the same number of $u_i$ and $v_j$. Here, it is not necessary the case. First of all, we have to define this partial Cauchy matrix and compute its inverse. We name it as partial in reference to the pDWBC partition function.

\subsubsection{The partial Cauchy matrix and its inverse}

To define the partial Cauchy matrix, we use the same form as the one of the partition function because the real Cauchy matrix corresponds to the matrix inside the determinant for $\beta = 0$ (in the case of a square lattice).

\begin{Def}
 We define the partial Cauchy matrix $\mathcal{C}_{nm}(\bu\vert\bv)$ by
 \begin{equation}
  \mathcal{C}_{nm}(\bu\vert\bv) = \left\{ 
  \begin{array}{cc}
     \left( \left\{ \begin{array}{cc} 
             g(u_i-v_j) & \text{if} \; i \leqslant n \\
             \psi_{m-i}(-v_j) & \text{if} \; i \geqslant n+1
       \end{array} \right. \right)_{1 \leqslant i,j \leqslant m}   & \text{for} \; n \leqslant m \\
       \left( \left\{ \begin{array}{cc} 
             g(u_i-v_j) & \text{if} \; j \leqslant m \\
             \psi_{n-j}(u_i) & \text{if} \; j \geqslant m+1
       \end{array} \right. \right)_{1 \leqslant i,j \leqslant n}
       & \text{for} \; n \geqslant m
   \end{array} \right.
 \end{equation}
\end{Def}

\begin{prop}
 The determinant of the partial Cauchy matrix $\mathcal{C}_{nm}(\bu\vert\bv)$ has the same form as the one of the real Cauchy matrix 
 \begin{equation}
  \mathrm{det}\left(\mathcal{C}_{nm}(\bu\vert\bv)\right) = g(\bu,\bv) \Delta(\bu) \Delta'(\bv).
 \end{equation}
\end{prop}

\begin{proof}
 The proof has the construction as the one of the formula \ref{newformulaZ}. We will focus on the case $n \leqslant m$. We obtain the other case exchanging $n \leftrightarrow m$, $u_i \leftrightarrow -v_j$ and $i \leftrightarrow j$. 
 
 We will prove the following statement by induction over $k \in [\![0,m-n]\!]$ 
 \begin{equation}\label{detCauchypartk}
  \mathrm{det}\left(\mathcal{C}_{m-k,m}(\bu\vert\bv)\right) = g(\bu,\bv) \Delta(\bu) \Delta'(\bv) 
 \end{equation}

 For $k = 0$, the statement is automatically true because the partial Cauchy matrix is the real Cauchy matrix and its determinant is well known : it is $g(\bu,\bv) \Delta(\bu) \Delta'(\bv)$. The interested reader can find a proof, for example, in the french book \cite{Gourdon2021}.
 
 We suppose now that the statement is true for one $k \in [\![0,m-n-1]\!]$ and we want to prove the statement for $k+1$. So we have to eliminate the variable $u_{m-k}$ : we will do this tending it toward infinity. 
 So, according to the induction statement,
 \begin{align}
  1 &= \frac{1}{g(\bu,\bv) \Delta(\bu) \Delta'(\bv)} \mathrm{det}\left(\mathcal{C}_{m-k,m}(\bu\vert\bv)\right) \\
  &= \frac{1}{g(\bar u_{m-k},\bv) \Delta(\bar u_{m-k}) \Delta'(\bv)} \left[ \prod_{j=1}^m \frac{u_{m-k}-v_j}{c} \prod_{i=1}^{m-k-1} \frac{c}{u_{m-k}-u_i} \right] \mathrm{det}\left(\mathcal{C}_{m-k,m}(\bu\vert\bv)\right).
 \end{align}
 We make the change of variable $u_{m-k} = \frac{1}{\varepsilon}$ and we inject the factor between brackets in the row $m-k$ of the determinant. The coefficients $c_{m-k,\tilde \jmath}$, for $\tilde \jmath \in [\![1,m]\!]$, become 
 \begin{equation}
  c_{m-k,\tilde \jmath} = \frac{1}{(\varepsilon c)^{k}} \frac{1}{1 - v_{\tilde \jmath} \varepsilon} \prod_{j=1}^m (1- v_j \varepsilon) \prod_{i=1}^{m-k-1} \frac{1}{1-u_i \varepsilon}.
 \end{equation}
 Using again $e_l$ and $h_l$ and their generating functions, we obtain
 \begin{equation}
  c_{m-k,\tilde \jmath} = \sum_{l=0}^{\infty} \left( \sum_{l_1+l_2+l_3=l} (-1)^{l_1}e_{l_1}\left( \frac{\bv}{c} \right) h_{l_2}\left(\frac{\bu_{m-k}}{c}\right) \left( \frac{v_{\tilde \jmath}}{c} \right)^{l_3} \right) (\varepsilon c)^{l-k}.
 \end{equation}
 Doing the same reasoning that in the previous proof, we can write the row $m-k$ with the term $\left( \frac{v_{\tilde \jmath}}{c} \right)^{k}$ for all $\tilde \jmath \in [\![1,m]\!]$ when $\varepsilon \rightarrow 0$. So we obtain 
 \begin{equation}
  1 = \frac{1}{g(\bu_{m-k},\bv) \Delta(\bu_{m-k}) \Delta'(\bv)} \mathrm{det}\left(\mathcal{C}_{m-(k+1),m}(\bu\vert\bv)\right), 
 \end{equation}
 because $1$ is independent of $u_{m-k}$, which proves the statement for $k+1$.
 
 Finally, by induction over $k \in [\![0,m-n]\!]$, 
 \begin{equation}
  \mathrm{det}\left(\mathcal{C}_{m-k,m}(\bu\vert\bv)\right) = g(\bu,\bv) \Delta(\bu) \Delta'(\bv).
 \end{equation}
For $k = m-n$, we obtain 
\begin{equation}
  \mathrm{det}\left(\mathcal{C}_{nm}(\bu\vert\bv)\right) = g(\bu,\bv) \Delta(\bu) \Delta'(\bv).
 \end{equation}
\end{proof}

\begin{prop}
 The inverse matrix of the partial Cauchy matrix $\mathcal{C}_{nm}(\bu\vert\bv)$ is 
 \begin{equation}
   \mathcal{C}_{nm}(\bu\vert\bv)^{-1} = \left\{ 
        \begin{array}{cc}
            \left( \begin{matrix}
             \left\{ \begin{array}{cc}
             g(u_j,v_i) \frac{g(\bu_j,u_j)g(v_i,\bv_i)}{g(u_j,\bv)g(\bu,v_i)} & \text{if} \; j \leqslant n \\
             \frac{g(v_i,\bv_i)}{g(v_i,\bu)} f_{j-(n+1)}(\bu\vert -\bv_i) & \text{if} \; j \geqslant n +1
             \end{array} \right.
            \end{matrix} \right) & \text{if} \;\; n \leqslant m \\
            \left( \begin{matrix}
             \left\{ \begin{array}{cc}
             g(u_j,v_i) \frac{g(\bu_j,u_j)g(v_i,\bv_i)}{g(u_j,\bv)g(\bu,v_i)} & \text{if} \; i \leqslant m \\
             \frac{g(\bu_j,u_j)}{g(\bv,u_j)} f_{i-(m+1)}(-\bv\vert \bu_j) & \text{if} \; i \geqslant m+1
             \end{array} \right.
            \end{matrix} \right)& \text{if} \;\; n \geqslant m
        \end{array}
   \right.
 \end{equation}
 where 
 \begin{equation}
     f_p(\bar{x}|\bar{y}) = \frac{1}{c^p} \sum_{l=0}^p h_{p-l}(\bar{x})e_l(\bar{y}).
 \end{equation}
\end{prop}

\begin{proof}
 We denote $\mathcal{IC}_{nm}(\bu\vert\bv)$ the previous matrix. We begin with the case $n \leqslant m$ and we will compute the coefficient $m^{(1)}_{ij}$ of the matrix $\mathcal{C}_{nm}(\bu\vert\bv)\mathcal{IC}_{nm}(\bu\vert\bv)$. So we will sum on the elements of $\bv$.
 \paragraph{For $i,j \in [\![1,n]\!]$,}
    \begin{align}
    m^{(1)}_{ij} &= \sum_{k=1}^m \frac{g(u_i,v_k)g(u_j,v_k)g(\bu_j,u_j)g(v_k,\bv_k)}{g(u_j,\bv)g(\bu,v_k)} \\
    &=\frac{g(\bu_j,u_j)}{g(u_j,\bv)} \sum_{k=1}^m \frac{g(u_i,v_k)g(u_j,v_k)g(v_k,\bv_k)}{g(\bu,v_k)}.
    \end{align}
    We denote the complex rational function 
    \begin{equation}
        F_{ij}(z) = \frac{g(u_i,z)g(u_j,z)g(z,\bv)}{g(\bu,z)}.
    \end{equation}
    The degree of the denominator is $m+2$, the one of the numerator is $n$ and we have $m+2 > n +1$.
    \begin{itemize}
        \item If $i \neq j$, the poles of $F_{ij}$ are all the elements of $\bv$, they are of order $1$ and their residue is 
        \begin{equation}
            \mathrm{Res}(F_{ij},v_k) = \frac{g(u_i,v_k)g(u_j,v_k)g(v_k,\bv_k)}{g(\bu,v_k)c}.
        \end{equation}
        So, using the proposition \ref{calculsum}, we obtain 
        \begin{equation}
            \sum_{k=1}^m \frac{g(u_i,v_k)g(u_j,v_k)g(v_k,\bv_k)}{g(\bu,v_k)} = 0.
        \end{equation}
        \item If $i = j$, the poles of $F_{jj}$ are all the elements of $\bv$, at the same order and with the same residue, and $u_j$, at the order $1$ and with the residue
        \begin{equation}
            \mathrm{Res}(F_{jj},u_j) = -\frac{1}{c}\frac{g(u_j,\bv)}{g(\bu_j,u_j)}
        \end{equation}
        So, using the proposition \ref{calculsum},
        \begin{equation}
            \sum_{k=1}^m \frac{g(u_i,v_k)g(u_j,v_k)g(v_k,\bv_k)}{g(\bu,v_k)} = \frac{g(u_j,\bv)}{g(\bu_j,u_j)}.
        \end{equation}
    \end{itemize}
    Therefore, $m^{(1)}_{ij} = \delta_{ij}$.
    
    \paragraph{For $i \in [\![n+1,m]\!]$ and $j \in [\![1,n]\!]$,}
    \begin{align}
     m^{(1)}_{ij} &= \sum_{k=1}^m \left( \frac{v_k}{c} \right)^{m-i} \frac{g(u_j,v_k)g(\bu_j,u_j)g(v_k,\bv_k)}{g(u_j,\bv)g(\bu,v_k)} \\
     &= \frac{g(\bu_j,u_j)}{g(u_j,\bv)} \sum_{k=1}^m \left( \frac{v_k}{c} \right)^{m-i} \frac{g(u_j,v_k)g(v_k,\bv_k)}{g(\bu,v_k)}.
    \end{align}
    We denote the complex rational function
    \begin{equation}
     F_{ij}(z) = \left( \frac{z}{c} \right)^{m-i} \frac{g(u_j,z)g(z,\bv)}{g(\bu,z)}.
    \end{equation}
    The degree of the denominator is $m+1$, the one of the numerator is $n+m-i$ and $m+1 > n+m-i+1$ (because $i > n$). 
    The poles of $F_{ij}$ are all the elements of $\bv$, they are of order $1$ and their residue is 
    \begin{equation}
        \mathrm{Res}(F_{ij},v_k) = \left( \frac{v_k}{c} \right)^{m-i} \frac{g(u_j,v_k)g(v_k,\bv_k)}{g(\bu,v_k) c}.
    \end{equation}
    So, using the proposition \ref{calculsum}, we obtain 
    \begin{equation}
        \sum_{k=1}^m \left( \frac{v_k}{c} \right)^{m-i} \frac{g(u_j,v_k)g(v_k,\bv_k)}{g(\bu,v_k)} = 0.
    \end{equation}
    
    \paragraph{For $i \in [\![1,n]\!]$ and $j \in [\![n+1,m]\!]$,}
    \begin{align}
     m^{(1)}_{ij} &= \sum_{k=1}^m g(u_i,v_k)\frac{g(v_k,\bv_k)}{g(v_k,\bu)}f_{j-(n+1)}(\bu\vert-\bv_k) \\
     &= \sum_{k=1}^m g(u_i,v_k)\frac{g(v_k,\bv_k)}{g(v_k,\bu)} \frac{1}{c^{j-n-1}} \sum_{l=0}^{j-n-1} (-1)^lh_{j-n-1-l}(\bu)e_l(\bv_k) \\
     &= \frac{-1}{c^{j-n-1}} \sum_{l=0}^{j-n-1} h_{j-n-1-l}(\bu) \sum_{k=1}^m (-1)^l e_l(\bv_k) \frac{g(v_k,\bv_k)}{g(v_k,\bu_i)} .
    \end{align}
    And
    \begin{equation}
        \sum_{k=1}^m (-1)^l e_l(\bv_k)  \frac{g(v_k,\bv_k)}{g(v_k,\bu_i)}
        = c^{m-n} \sum_{k=1}^m (-1)^l e_l(\bv_k) \underset{p \neq i}{\prod_{p = 1}^n} (v_k-u_p)  \underset{q \neq k}{\prod_{q = 1}^m} \frac{1}{v_k - v_q}  
    \end{equation}
    with, using the generating function of $e_p$ \eqref{genersymmfunc},
    \begin{equation}
        \underset{p \neq i}{\prod_{p = 1}^n} (v_k-u_p) = v_k^{n-1} E_{-\bu_i}(v_k^{-1}) = \sum_{p = 0}^{n-1} (-1)^p e_p(\bu_i) v_k^{n-1-p}.
    \end{equation}
    Therefore, 
    \begin{equation}
        \sum_{k=1}^m (-1)^l e_l(\bv_k)\frac{g(v_k,\bv_k)}{g(v_k,\bu_i)} 
        = c^{m-n} \sum_{p = 0}^{n-1} (-1)^p e_p(\bu_i) \sum_{k=1}^m v_k^{n-1-p}  (-1)^l e_l(\bv_k) \underset{q \neq k}{\prod_{q = 1}^m} \frac{1}{v_k - v_q}.
    \end{equation}
    We can compute the second sum in the right term using the equation \eqref{eqVand}, with $N = m$, $N-i = l$ and $j = n-p$ :
    \begin{equation}
        \sum_{k=1}^m v_k^{n-1-p} (-1)^l e_l(\bv_k) \underset{q \neq k}{\prod_{q = 1}^m} \frac{1}{v_k - v_q}  = \delta_{m-l,n-p}.
    \end{equation}
    So 
    \begin{equation}
     m^{(1)}_{ij} = \frac{-1}{c^{j-m-1}} \sum_{l=0}^{j-n-1} h_{j-n-1-l}(\bu) \sum_{p = 0}^{n-1} e_p(-\bu_i) \delta_{m-l,n-p}.
    \end{equation}
    However, $0 \leqslant p \leqslant n-1$ so $1 \leqslant n-p \leqslant n$ and $0 \leqslant l \leqslant j-n-1 \leqslant m-n-1$ so $n+1 \leqslant m-l \leqslant m$. So, $\delta_{m-l,n-p} = 0$, for all $p$ and $l$, and $m^{(1)}_{ij} = 0$.
    
    \paragraph{For $i,j \in [\![n+1,m]\!]$,}
    \begin{align}
        m^{(1)}_{ij} &= \sum_{k=1}^m \left( \frac{v_k}{c} \right)^{m-i} \frac{g(v_k,\bv_k)}{g(v_k,\bu)}f_{j-(n+1)}(\bu\vert-\bv_k) \\
        &= \frac{1}{c^{j-n-1}} \left( \frac{v_k}{c} \right)^{m-i}\sum_{l=0}^{j-n-1} h_{j-n-1-l}(\bu) \sum_{k=1}^m (-1)^l e_l(\bv_k) \frac{g(v_k,\bv_k)}{g(v_k,\bu)}\\
        &= \frac{1}{c^{j-i}} \sum_{l=0}^{j-n-1} h_{j-n-1-l}(\bu) \sum_{p=0}^n (-1)^p e_p(\bu) \sum_{k=1}^m v_k^{m-i+n-p} (-1)^l e_l(\bv_k) \underset{q \neq k}{\prod_{q=1}^m} \frac{1}{v_k-v_q} 
    \end{align}
    using the same computations as just before. And, as just before also, 
    \begin{equation}
        \sum_{k=1}^m v_k^{m-i+n-p} (-1)^l e_l(\bv_k) \underset{q \neq k}{\prod_{q=1}^m} \frac{1}{v_k-v_q} = \delta_{m-l,m-i+n-p} = \delta_{p,n+1+l-i}.
    \end{equation}
    Therefore, 
    \begin{align}
        m^{(1)}_{ij} &= \frac{1}{c^{j-i}} \sum_{l=0}^{j-n-1} h_{j-n-1-l}(\bu) \sum_{p=0}^n (-1)^p e_p(\bu) \delta_{p,n+1+l-i} \\
        &= \frac{1}{c^{j-i}} \sum_{l=0}^{j-n-1} (-1)^{n+1+l-i} h_{j-n-1-l}(\bu) e_{n+1+l-i}(-\bu) \mathds{1}_{0 \leqslant n+1+l-i \leqslant n} \\
        &= \frac{1}{c^{j-i}} \sum_{l=0}^{j-n-1} (-1)^{j-i-l} h_{l}(\bu) e_{j-i-l}(-\bu) \mathds{1}_{0 \leqslant j-i-l \leqslant n}
    \end{align}
    making the change of index $l'=j-n-1-l$. In addition, we know that $e_{j-i-l}(-\bu) = 0$ if $j-i-l > n$ by definition of $e_r$ (see definition \ref{defsymmfunc}). So 
    \begin{equation}
        m^{(1)}_{ij} = \frac{1}{c^{j-i}} \sum_{l=0}^{j-n-1} h_{l}(\bu) e_{j-i-l}(-\bu) \mathds{1}_{l \leqslant j-i} = \frac{1}{c^{j-i}} \sum_{l=0}^{j-i} h_{l}(\bu) e_{j-i-l}(-\bu)
    \end{equation}
    because $j-i \leqslant j-n-1$.
    \begin{itemize}
        \item If $i > j$, the sum is empty and $m^{(1)}_{ij} = 0$.
        \item If $i = j$, the sum contains only one term : $m_{ii} = h_0(\bu)e_0(\bu) = 1$.
        \item If $i < j$, according to the equation \eqref{identitysymmfunc}, $m^{(1)}_{ij} = 0$.
    \end{itemize}
    So $m^{(1)}_{ij} = \delta_{ij}$.
 
 \paragraph{} To conclude this part, for all $i,j \in [\![1,m]\!]$, $m^{(1)}_{ij} = \delta_{ij}$ and $\mathcal{IC}_{nm}(\bu\vert\bv) = \mathcal{C}_{nm}(\bu\vert\bv)^{-1}$.
 
 The proof for the case $n \leqslant m$ is similar to the previous case. 
 This time, we compute the coefficient of the matrix $\mathcal{IC}_{nm}(\bu\vert\bv)\mathcal{C}_{nm}(\bu\vert\bv)$ (we sum on the elements of $\bu)$. 
 The following part is indentical to the previous case, exchanging $n \leftrightarrow m$, $i \leftrightarrow j$ and $\bu \leftrightarrow -\bv$.
 \end{proof}
 
 \begin{remark}
  The matrices
  \begin{equation}
    \tilde{\mathcal{C}}_{nm}(\bu\vert\bv) = \left\{ 
  \begin{array}{cc}
     \left( \left\{ \begin{array}{cc} 
             h(u_i-v_j)^{-1} & \text{if} \; i \leqslant n \\
             \psi_{m-i}(-v_j) & \text{if} \; i \geqslant n+1
       \end{array} \right. \right)_{1 \leqslant i,j \leqslant m}   & \text{for} \; n \leqslant m \\
       \left( \left\{ \begin{array}{cc} 
             h(u_i-v_j)^{-1} & \text{if} \; j \leqslant m \\
             \psi_{n-j}(u_i) & \text{if} \; j \geqslant m+1
       \end{array} \right. \right)_{1 \leqslant i,j \leqslant n}
       & \text{for} \; n \geqslant m
   \end{array} \right.
  \end{equation}
 corresponds also to partial Cauchy matrices : $\mathcal{C}_{nm}(\bu+c\vert\bv)$ in the first case and $\mathcal{C}_{nm}(\bu\vert\bv-c)$ in the second one because 
 \begin{equation}
  h(u,v) = g(u+c,v)^{-1} = g(u,v-c)^{-1}.
 \end{equation}
 We translate by $\pm c$, in the first case, the elements of $\bu$ and, in the second case, the elements of $\bv$ because we do not want to put $c$ in the Vandermonde part of the matrix (but it is possible). Its determinant therefore is
 \begin{equation}
  \mathrm{det} \left( \tilde{\mathcal{C}}_{nm}(\bu\vert\bv) \right) = \frac{\Delta(\bu)\Delta'(\bv)}{h(\bu,\bv)}
 \end{equation}
 and its inverse is
 \begin{equation}\label{inverseCauchyparttrans}
   \tilde{\mathcal{C}}_{nm}(\bu\vert\bv)^{-1} = \left\{ \begin{array}{cc}
        \mathcal{C}_{nm}(\bu+c\vert\bv)^{-1} & \text{if} \;\; n \leqslant m \\
        \mathcal{C}_{nm}(\bu\vert\bv-c)^{-1} & \text{if} \;\; n \geqslant m
        \end{array}
   \right. .
 \end{equation}
 \end{remark}

 \subsubsection{Transformations from the new formula to known ones}
 
 \paragraph{For $n \leqslant m$} We start from the formula \eqref{newformulaZ} and we will obtain the formula \eqref{ZwithK}-\eqref{K2}. 
We consider therefore that $g(\bu,\bv)\Delta(\bu)\Delta'(\bv)$ is the determinant of the partial Cauchy matrix $\mathcal{C}_{nm}(\bu\vert\bv)$. 
So we inject this expression in the determinant of $Z_{nm}(\bu|\bv|B|\hat C)$, multiplying the matrix inside by $\mathcal{C}_{nm}(\bu\vert\bv)^{-1}$ from the right (we sum over the elements of $\bv$). We compute the obtained coefficient $m^{(2)}_{ij}$.
\subparagraph{For $i,j \in [\![1,n]\!]$,}
 \begin{align}
  m^{(2)}_{ij} &= \sum_{k=1}^m \left( g(u_i,v_k) - \frac{\beta}{h(u_i,v_k)} \right) g(u_j,v_k) \frac{g(\bu_j,u_j)g(v_k,\bv_k)}{g(u_j,\bv)g(\bu,v_k)} \\
  &= m^{(1)}_{ij} - \beta \frac{g(\bu_j,u_j)}{g(u_j,\bv)}\sum_{k=1}^m  \frac{g(u_j,v_k)g(v_k,\bv_k)}{h(u_i,v_j)g(\bu,v_k)}.
 \end{align}
However, according to the previous part, $m^{(1)}_{ij} = \delta_{ij}$. And we denote the complex rational function 
\begin{equation}
 F_{ij}(z) = \frac{g(u_j,z)g(z,\bv)}{h(u_i,v_k)g(\bu,z)}.
\end{equation}
The degree of the denominator is $m+2$, the one of the numerator is $n$ and $m+2 > n+1$. Its poles are all the elements of $\bv$ and $u_i+c$, they are of order $1$ and their residue is
\begin{equation}
 \mathrm{Res}(F_{ij},v_k) = \frac{1}{c} \frac{g(u_j,v_k)g(v_k,\bv_k)}{h(u_i,v_k)g(\bu,v_k)},
\end{equation}
\begin{equation}
 \mathrm{Res}(F_{ij},u_i+c) = -\frac{1}{c} \frac{g(u_j,u_i+c)g(u_i+c,\bv)}{g(\bu,u_i+c)} = \frac{(-1)^n}{c} \frac{h(u_i,\bu)}{h(u_i,u_j)h(u_i,\bv)}.
\end{equation}
So, using the proposition \ref{calculsum},
\begin{equation}
 \sum_{k=1}^m  \frac{g(u_j,v_k)g(v_k,\bv_k)}{h(u_i,v_j)g(\bu,v_k)} = - (-1)^n \frac{h(u_i,\bu)}{h(u_i,u_j)h(u_i,\bv)} = (-1)^{n-1} \frac{h(u_i,\bu)}{h(u_i,u_j)h(u_i,\bv)}
\end{equation}
and
\begin{equation}
  m^{(2)}_{ij} = \delta_{ij} - \beta \frac{(-1)^{n-1}g(\bu_j,u_j)}{g(u_j,\bv)} \frac{h(u_i,\bu)}{h(u_i,u_j)h(u_i,\bv)} = \delta_{ij} - \beta \frac{h(u_i,\bu)g(u_j,\bu_j)}{h(u_i,u_j)g(u_j,\bv)h(u_i,\bv)}.
\end{equation}

\subparagraph{For $i \in [\![n+1,m]\!]$ and $j \in [\![1,n]\!]$,}
\begin{equation}
 m^{(2)}_{ij} = \sum_{k=1}^m \left( \frac{v_k}{c} \right)^{m-i} \frac{g(u_j,v_k)g(\bu_j,u_j)g(v_k,\bv_k)}{g(u_j,\bv)g(\bu,v_k)} = m^{(1)}_{ij} = 0.
\end{equation}

 \subparagraph{For $i,j \in [\![n+1,m]\!]$,}
\begin{equation}
 m^{(2)}_{ij} = \sum_{k=1}^m \left( \frac{v_k}{c} \right)^{m-i} \frac{g(v_k,\bv_k)}{g(v_k,\bu)}f_{j-(n+1)}(\bu\vert-\bv_k) = m^{(1)}_{ij} = \delta_{ij}.
\end{equation}

\subparagraph{For $i \in [\![n+1,m]\!]$ and $j \in [\![1,n]\!]$,} we do not need to compute this last block.

\subparagraph{} The obtained matrix is therefore an upper triangular one 
\begin{equation}
 \left( \begin{array}{c|c}
         \delta_{ij} - \beta \frac{g(u_j,\bu_j)h(u_i,\bu)}{h(u_i,u_j)g(u_j,\bv)h(u_i,\bv)} & ??? \\
         \hline
         (0) & \delta_{ij} 
        \end{array} \right).
\end{equation}
and
$Z_{nm}(\bu|\bv|B|\hat C)$ has the form $\frac{\mathrm{tr}(B)^n\mathrm{tr}(\hat C)^m}{(1-\beta)^n} h(\bu,\bv) \underset{n}{\mathrm{det}}\left( I_n - \beta M_1(\bu\vert\bv)\right)$ with 
\begin{equation}
     M_1(\bu\vert\bv)_{ij} = \frac{g(u_j,\bu_j)h(u_i,\bu)}{h(u_i,u_j)g(u_j,\bv)h(u_i,\bv)}.
\end{equation}
Now, we will transform the matrix inside the determinant to obtain the wanted formula. We remark that 
\begin{equation}
    \frac{g(u_j,\bu_j)h(u_i,\bu)}{h(u_i,u_j)g(u_j,\bv)h(u_i,\bv)} = \frac{g(u_j,\bu_j)}{g(u_j,\bv)} \frac{f(u_i,\bu_i)}{f(u_i,\bv)h(u_i,u_j)} \frac{g(u_i,\bv)}{g(u_i,\bu_i)}
\end{equation}
(with $h(u,u) = 1$) so, using the determinant property $\mathrm{det}(AB) = \mathrm{det}(A)\mathrm{det}(B)$ and the equality
\begin{equation}
    I_n - \beta M_1(\bu\vert\bv) = D_1(\bu\vert\bv) ( I_n - \beta M_2(\bu\vert\bv)) D_1(\bu\vert\bv)^{-1},
\end{equation}
with 
\begin{equation}
 M_2(\bu\vert\bv)_{ij} = \frac{f(u_i,\bu_i)}{f(u_i,\bv)h(u_i,u_j)} \;\;\; \text{and} \;\;\; D_1(\bu\vert\bv)_{ij} = \delta_{ij} \frac{g(u_j,\bu_j)}{g(u_j,\bv)},
\end{equation}
we obtain
\begin{equation}
    Z_{nm}(\bu|\bv|B|\hat C) = \frac{\mathrm{tr}(B)^n\mathrm{tr}(\hat C)^m}{(1-\beta)^n} h(\bu,\bv) \underset{n}{\mathrm{det}}\left( I_n - \beta M_2(\bu\vert\bv)\right).
\end{equation}
So, injecting the factor $f(\bu,\bv)$ in the determinant, we obtain the first equality \eqref{ZwithK}-\eqref{K2}.
 
 \paragraph{For $n \geqslant m$} We start also from the formula \eqref{newformulaZ} and we will obtain this time the formula  \eqref{ZwithK}-\eqref{K1}. 
We consider this time that $\Delta(\bu)\Delta'(\bv)/h(\bu,\bv)$ is the determinant of the partial Cauchy matrix $\tilde{\mathcal{C}}_{nm}(\bu\vert\bv)$. We inject this expression in the determinant of $Z_{nm}(\bu|\bv|B|\hat C)$, multiplying this time the matrix inside by $\tilde{\mathcal{C}}_{nm}(\bu\vert\bv)^{-1}$ \eqref{inverseCauchyparttrans} from the left (we sum over the elements of $\bu$). Doing the same reasoning as before, we obtained a lower triangular matrix 
\begin{equation}
 \left( \begin{array}{c|c}
         - \beta \delta_{ij} + \frac{g(\bu,v_j)h(\bu,v_i)g(v_i,\bv_i)h(v_j,\bv)}{h(v_j,v_i)} & (0) \\
         \hline 
         ??? & \delta_{ij}
        \end{array} \right)
\end{equation}
and 
\begin{equation}
 Z_{nm}(\bu|\bv|B|\hat C) = \frac{\mathrm{tr}(B)^n\mathrm{tr}(\hat C)^m}{(1-\beta)^m} \frac{1}{g(\bu,\bv)}
\underset{1 \leqslant i,j \leqslant m}{\mathrm{det}} \left( - \beta \delta_{ij} + \frac{g(\bu,v_j)h(\bu,v_i)g(v_i,\bv_i)h(v_j,\bv)}{h(v_j,v_i)} \right).
\end{equation}
Finally, remarking that 
\begin{equation}
 \frac{g(\bu,v_j)h(\bu,v_i)g(v_i,\bv_i)h(v_j,\bv)}{h(v_j,v_i)} = h(\bu,v_i)g(v_i,\bv_i) \frac{f(\bu,v_j)f(v_j,\bv_j)}{h(v_j,v_ii)} \frac{1}{h(\bu,v_j)g(v_j,\bv_j)}
\end{equation}
and doing the same type of transformation as before, we obtain the second equality \eqref{ZwithK}-\eqref{K1}.

\paragraph{} The figure \ref{fig:equality3formulae} summarizes the different ways of proofs between the three formulae of the modified Izergin determinant.

\begin{figure}[ht]
    \centering
    \includegraphics[width=\textwidth]{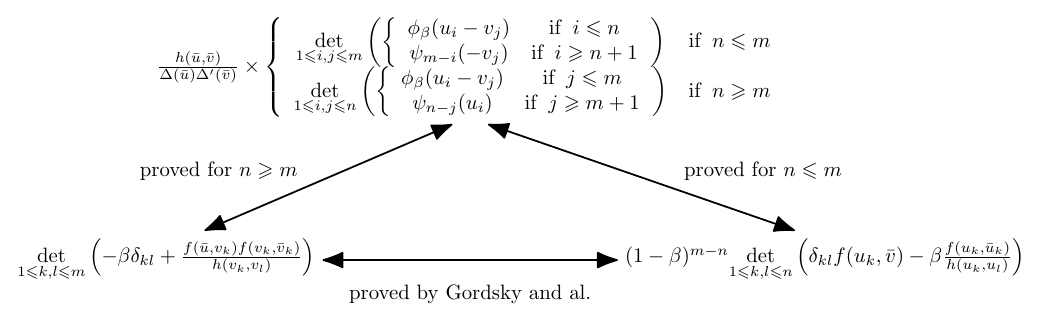}
    \caption{\centering Ways of proofs between the three formulae for $K_{nm}^{(\beta)}(\bu|\bv)$. For the down arrow, see the proposition 4.1 of \cite{Gorsky2014}}
    \label{fig:equality3formulae}
\end{figure}

\section{Correspondence with the pDWBC model}\label{pDWBC}

In \cite{Foda2012}, the authors defined the partial domain wall boundary conditions for a $n \times m$ rectangular lattice, with $n \leqslant m$. They used another representation for the lattice : they used arrows instead of spins on the edges of the lattice. 

First, we will explicit this representation to avoid any ambiguity. There exists a correspondence between spins and arrows on edges (as we see in the figure \ref{fig:correspondance}) : an up spin corresponds to an arrow in the same direction of the orientation of the line and a down spin corresponds to an arrow in the opposite direction of the orientation of the line.
So we can represent the six different vertices with arrows as in the figure \ref{fig:six_vertex_arrow}. We observe that all the vertices conserve the arrow flow : two arrows point toward the vertex and two point away from it. Except for our convention, we find the same vertices.

\begin{figure}[ht]
 \centering
 \begin{subfigure}[h]{0.25\textwidth}
  \centering
  \includegraphics[width = \textwidth]{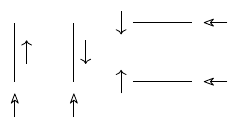}
  \caption{\centering Spins}
  \label{fig:correspondance_spin}
 \end{subfigure}
 \hspace{15mm}
 \begin{subfigure}[h]{0.25\textwidth}
  \centering
  \includegraphics[width = \textwidth]{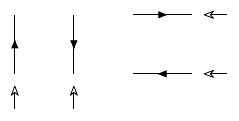}
  \caption{\centering Arrows}
  \label{fig:correspondance_arrows}
 \end{subfigure}
 \caption{\centering Correspondance between both representations.}
 \label{fig:correspondance}
\end{figure}

\begin{figure}[ht]
 \centering
 \includegraphics[width=.7\textwidth]{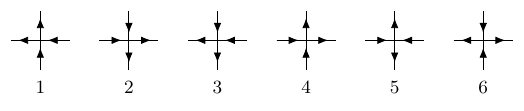}
 \caption{\centering The six different vertices with arrows}
 \label{fig:six_vertex_arrow}
\end{figure}

Using our notations and our convention, we recall the definition and the expression of the corresponding partition function given in \cite{Foda2012}.

\begin{Def}\label{defpDWBC}
 The partial domain wall boundary conditions are defined as follows.
 \begin{enumerate}
  \item All arrows on the western and eastern boundaries point inwards.
  \item $k \in [\![0,m-n]\!]$ arrows on the northern boundary and $m-n-k$ arrows on the southern boundary also point inwards.
  \item The remaining $m+n$ arrows on the northern and southern boundaries point outwards.
  \item With $k$ fixed, the locations of the inward-pointing arrows on the northern and southern boundaries are summed over.
 \end{enumerate}
\end{Def}

\begin{remark}
 To write more explicitly the point 4, for one $k \in [\![0,m-n]\!]$, we consider all the different possible configurations of the lattice, included the northern and southern boundaries, and we sum all the expressions to obtain $Z_{nm}(\bu\vert\bv\vert k)$.
\end{remark}
 
\begin{prop}\label{ZpDWBC}
 The partition function of the inhomogeneous six-vertex model with partial domain wall boundary conditions, considering that the northern boundary counts $k \in [\![0,m-n]\!]$ arrows pointing inwards, is 
 \begin{equation}
    Z_{nm}(\bar u \vert \bar v \vert k) = \binom{m-n}{k} \frac{h(\bu,\bv)}{g(\bu,\bv)\Delta(\bu)\Delta'(\bv)} 
    \underset{1 \leqslant i,j \leqslant m}{\mathrm{det}} \left( \left\{ \begin{array}{cc}
               \phi_1(u_i-v_j) & \text{if} \;\; i \leqslant n \\
                \psi_{m-i}(-v_j) & \text{if} \;\; i \geqslant n+1
           \end{array} \right. \right)
 \end{equation}
\end{prop}

\begin{remark}
 The formula for $k = m-n$, also equal to $Z_{nm}(\bu\vert\bv\vert 0)$, is obtained at the equation (3.14) of \cite{Foda2012} and the binomial coefficient before is given at the end of the subsection 2.11.
\end{remark}

Now, we will show that our formula \eqref{newformulaZ} generalizes the previous result of \cite{Foda2012}.
We consider the inhomogeneous six-vertex model on the $n \times m$ rectangular lattice with the general boundary conditions defined in the figure \ref{fig:bords_generaux}. We fix the western and eastern boundaries like in the definition \ref{defpDWBC}, all this arrows point inward, but we do not fix the northern and southern boundaries. So, with our convention and the equivalence between spins and arrows,

\begin{equation}
    \bra{n} = \begin{pmatrix} n_1 & n_2 \end{pmatrix}, \;\;
    \ket{s} = \begin{pmatrix} s_1 \\ s_2 \end{pmatrix}, \;\;
    \ket{e} = \ket{1} = \begin{pmatrix} 1 \\ 0 \end{pmatrix} \;\; \text{and} \;\; \bra{w} = \bra{2} = \begin{pmatrix} 0 & 1 \end{pmatrix}.
\end{equation}
So 
\begin{equation}
 \mathrm{tr}(B) = \braket{w\vert e} = 0, \;\; \mathrm{tr}(\hat C) = \braket{n \vert s} = n_1s_1+n_2s_2, \;\; \mathrm{tr}(B \hat C) = \braket{n\vert e}\braket{w\vert s} = n_1s_2
\end{equation}
and $\beta = 1$. Therefore, we can write the new formula for the partition function as 
\begin{equation}
 Z_{nm}(\bu\vert\bv\vert B \vert \hat C) = \mathrm{tr}(B \hat C)^n\mathrm{tr}(\hat C)^{m-n} Z_{nm}(\bu\vert\bv\vert 0).
\end{equation}
We can develop the prefactor :
\begin{align}
 \mathrm{tr}(B \hat C)^n\mathrm{tr}(\hat C)^{m-n} &= (n_1s_2)^n(n_1s_1+n_2s_2)^{m-n} \\
 &= (n_1s_2)^n \sum_{k=0}^{m-n} \binom{m-n}{k} (n_1s_1)^{m-n-k} (n_2s_2)^k \\
 &= \sum_{k=0}^{m-n} \binom{m-n}{k} n_1^{m-k} n_2^k s_1^{m-n-k} s_2^{n+k}.
\end{align}
So
\begin{align}
 Z_{nm}(\bu\vert\bv\vert B \vert \hat C) &= \sum_{k=0}^{m-n} \binom{m-n}{k} n_1^{m-k} n_2^k s_1^{m-n-k} s_2^{n+k} Z_{nm}(\bu\vert\bv\vert 0) \\
 &= \sum_{k=0}^{m-n} n_1^{m-k} n_2^k s_1^{m-n-k} s_2^{n+k} Z_{nm}(\bu\vert\bv\vert k).
\end{align}
The prefactor $n_1^{m-k} n_2^k s_1^{m-n-k} s_2^{n+k}$, for $k \in [\![0,m-n]\!]$, corresponds to the different possible forms of the northern and southern boundaries, consistent with the points 2 and 3 of the definition \ref{defpDWBC} of the partial domain wall boundary conditions. Indeed, for $k \in [\![0,m-n]\!]$, 
\begin{itemize}
 \item $n_2^k$ indicates that there are $k$ arrows pointing inwards the lattice (or, equivalently, $k$ downs spins) on the northern boundary,
 \item $n_1^{m-k}$ indicates that there are $m-k$ arrows pointing outwards the lattice (or, equivalently, $m-k$ up spins) on the northern boundary,
 \item $s_1^{m-n-k}$ indicates that there are $m-n-k$ arrows pointing inwards the lattice (or, equivalently $m-n-k$ up spins) on the southern boundary
 \item and $s_2^{n+k}$ indicates that there are $n+k$ arrows pointing outwards the lattice (or, equivalently $n+k$ down spins) on the southern boundary.
\end{itemize}
Therefore, the modified rational six-vertex model enables us to obtain all the different possible configurations of the rational six-vertex model with partial domain wall boundary conditions.

\section{Computations for the infinite lattice}\label{appendixC}

\subsection{A matrix with binomial coefficients}

\begin{Def}
 For $d \in \NN$ and $n \in \NN^*$, we define the matrix
 \begin{equation}
  M_n(d) := \left( \binom{d+i-1+j-1}{i-1} \right)_{1 \leqslant i,j \leqslant n}.
 \end{equation}
\end{Def}

\begin{prop}\label{propdet}
    For all $d \in \NN$ and $n \in \NN^*$,
    \begin{equation}
        D_n(d) := \mathrm{det}(M_n(d)) = 1.
    \end{equation}
\end{prop}

\begin{proof}
 We prove the following statement by induction over $n \in \NN^*$
 \begin{equation}
  \forall d \in \NN, \;\; D_n(d) = 1.
 \end{equation}

 For $n = 1$, $D_1(d) = \binom{d}{0} = 1$ so the base case is proved. We suppose now that the statement is true for one $n \in \NN^*$ and we consider $D_{n+1}(d)$. So we compute this quantity.
 \begin{equation}
  D_{n+1}(d) = \left \vert 
        \begin{matrix}
         \binom{d}{0} & \binom{d+1}{0} & \binom{d+2}{0} & \cdots & \binom{d+n}{0} \\
         \binom{d+1}{1} & \binom{d+2}{1} & \binom{d+3}{1} & \cdots & \binom{d+n+1}{1} \\
         \binom{d+2}{2} & \binom{d+3}{2} & \binom{d+4}{2} & \cdots & \binom{d+n+2}{2} \\
         \vdots & \vdots & \vdots & \ddots & \vdots \\
         \binom{d+n}{n} & \binom{d+n+1}{n} & \binom{d+n+2}{n} & \dots & \binom{d+2n}{n} \\
        \end{matrix}
\right \vert .
 \end{equation}
 We replace each column $C_j$ for $j \in [\![2,n+1]\!]$ by $C_j - C_{j-1}$ and the first row by its value
 \begin{equation}\label{propdetcalcul1}
  D_{n+1}(d) = \left \vert 
        \begin{matrix}
         1 & 0 & 0 & \cdots & 0 \\
         \binom{d+1}{1} & \binom{d+2}{1}-\binom{d+1}{1} & \binom{d+3}{1}-\binom{d+2}{1} & \cdots & \binom{d+n+1}{1}-\binom{d+n}{1} \\
         \binom{d+2}{2} & \binom{d+3}{2}-\binom{d+2}{2} & \binom{d+4}{2}-\binom{d+3}{2} & \cdots & \binom{d+n+2}{2}-\binom{d+n+1}{2} \\
         \vdots & \vdots & \vdots & \ddots & \vdots \\
         \binom{d+n}{n} & \binom{d+n+1}{n}-\binom{d+n}{n} & \binom{d+n+2}{n}-\binom{d+n+1}{n} & \dots & \binom{d+2n}{n}-\binom{d+2n-1}{n} \\
        \end{matrix}
\right \vert .
 \end{equation}
Using the Pascal's formula $\binom{p+1}{q+1} - \binom{p}{q+1} = \binom{p}{q}$,
\begin{equation}
  D_{n+1}(d) = \left \vert 
        \begin{matrix}
         1 & 0 & 0 & \cdots & 0 \\
         \binom{d+1}{1} & \binom{d+1}{0} &\binom{d+2}{0} & \cdots & \binom{d+n}{0} \\
         \binom{d+2}{2} & \binom{d+2}{1} & \binom{d+3}{1} & \cdots & \binom{d+n+1}{1} \\
         \vdots & \vdots & \vdots & \ddots & \vdots \\
         \binom{d+n}{n} & \binom{d+n}{n-1} & \binom{d+n+1}{n-1} & \dots & \binom{d+2n-1}{n-1} \\
        \end{matrix}
\right \vert .
 \end{equation}
Using the Laplace expansion along the first row and rewriting the binomial coefficient,
\begin{equation}
  D_{n+1}(d) = \left \vert 
        \begin{matrix}
         \binom{d+1}{0} &\binom{d+1+1}{0} & \cdots & \binom{d+1+n-1}{0} \\
         \binom{d+1+1}{1} & \binom{d+1+2}{1} & \cdots & \binom{d+1+n}{1} \\
         \vdots & \vdots & \ddots & \vdots \\
         \binom{d+1+n-1}{n-1} & \binom{d+1+n}{n-1} & \dots & \binom{d+1+2(n-1)}{n-1} \\
        \end{matrix}
\right \vert
 \end{equation}
so $D_{n+1}(d) = D_n(d+1) = 1$ using the statement for $n$. Therefore, the statement is true for $n+1$. So the statement is true for any $n \in \NN^*$.
\end{proof}

\begin{Def}
 For $d \in \NN$, $n \in \NN^*$ and $k,l \in [\![1,n]\!]$, we consider $D_{n-1}^{(k,l)}(d)$ the $(k,l)$ minor of $M_n(d)$. So its expression is
 \begin{equation}
  D_{n-1}^{(k,l)}(d) := \underset{1 \leqslant i,j \leqslant n-1}{\mathrm{det}} \left( \left\{ \begin{matrix}
                    \binom{d+i-1+j-1}{i-1} & \text{if} \; i < k \; \text{and} \; j < l \\
                    \binom{d+i-1+j}{i-1} & \text{if} \; i < k \; \text{and} \; j \geqslant l \\
                    \binom{d+i+j-1}{i} & \text{if} \; i \geqslant k \; \text{and} \; j < l \\
                    \binom{d+i+j}{i} & \text{if} \; i \geqslant k \; \text{and} \; j \geqslant l \\
                \end{matrix} \right. \right).
 \end{equation}
\end{Def}

\begin{remarks}
 For $n = 1$, the only defined determinant is $D_0^{(1,1)}$ which is a determinant of size $0$. By convention, we take it equal to $1$.
 
 In what follows, we could make appear some $D_{n-1}^{(k,l)}(d)$ with $k$ or $l$ that do not belong to the interval $[\![1,n]\!]$. In this case, we will consider that $D_{n-1}^{(k,l)}(d) = 0$ (because the corresponding minors does not exist) and we obtain the following extended definition.
\end{remarks}

The goal of what follows is to compute this quantity.
 
\begin{Def}\label{defDnkletendu}
 For $d \in \NN$, $n \in \NN^*$, $k \in \NN^*$ and $l \in \mathbb{Z}$, we define $D_{n-1}^{(k,l)}(d)$ by
 \begin{equation}
  D_{n-1}^{(k,l)}(d) := \left \{ \begin{array}{cc}
                              \underset{1 \leqslant i,j \leqslant n-1}{\mathrm{det}} \left( \left\{ \begin{matrix}
                    \binom{d+i-1+j-1}{i-1} & \text{if} \; i < k \; \text{and} \; j < l \\
                    \binom{d+i-1+j}{i-1} & \text{if} \; i < k \; \text{and} \; j \geqslant l \\
                    \binom{d+i+j-1}{i} & \text{if} \; i \geqslant k \; \text{and} \; j < l \\
                    \binom{d+i+j}{i} & \text{if} \; i \geqslant k \; \text{and} \; j \geqslant l \\
                \end{matrix} \right. \right) & \text{if } k,l \in [\![1,n]\!] \\
                & \\
                              0 & \text{else}
                             \end{array} \right. .
 \end{equation}
\end{Def}

\begin{prop}
 For $d \in \NN$, $n > 1$, $k > 1$ and $l \in \mathbb{Z}$, $D_{n-1}^{(k,l)}(d)$ satisfies the following inductive relation 
 \begin{equation}\label{relrecDnkl}
  D_{n-1}^{(k,l)}(d) = D_{n-2}^{(k-1,l-1)}(d+1) + D_{n-2}^{(k-1,l)}(d+1).
 \end{equation}
\end{prop}

\begin{proof}
 If $k > n$, $l > n$ or $l < 1$, each term of the previous relation is equal to $0$ because, in the first case, we have also $k-1 > n-1$, in the second case, we have also $l-1 > n-1$ and $l > n-1$ and in the last case, we have also $l-1 <1$.
 The inductive relation is therefore satisfied.

Now, we suppose that $1 < k \leqslant n$. First, we treat the case $n = 2$. We have, for all $d \in \NN$, $D_1^{(2,1)}(d) = D_1^{(2,2)}(d) = 1$ and, by convention, $D_0^{(1,1)}(d) = 1$. And, according to the extended definition \ref{defDnkletendu}, $D_0^{(1,2)}(d) = D_0^{(1,0)}(d) = 0$. So the inductive relation is automatically satisfied.

Now, we consider the case $n > 2$. To make the computations, we suppose also that $1 < l < n$. We will consider the cases $l=1$ and $l=n$ at the end of the proof.
To have this relation, we follow the same reasoning as in the proof of the proposition \ref{propdet}. We replace each column $C_j$ for $j \in [\![2,n-1]\!]$ by $C_j - C_{j-1}$ and we obtain a determinant with the same form as \eqref{propdetcalcul1} : each column $C_j$, for $j \in [\![2,n-1]\!]$ and $\neq l$, has now the form
\begin{equation}
 \left( \begin{matrix}
  0 \\
  \binom{d+1+j-1}{1} - \binom{d+1+j-1-1}{1} \\
  \vdots \\
  \binom{d+k-2+j-1}{k-2} - \binom{d+k-2+j-1-1}{k-2}\\
  \binom{d+k+j-1}{k} - \binom{d+k+j-1-1}{k}\\
  \vdots \\
  \binom{d+n-1+j-1}{n-1} - \binom{d+n-1+j-1-1}{n-1}
 \end{matrix} \right) \;\;\; \text{or} \;\;\; 
  \left( \begin{matrix}
  0 \\
  \binom{d+1+j}{1} - \binom{d+1+j-1}{1} \\
  \vdots \\
  \binom{d+k-2+j}{k-2} - \binom{d+k-2+j-1}{k-2}\\
  \binom{d+k+j}{k} - \binom{d+k+j-1}{k}\\
  \vdots \\
  \binom{d+n-1+j}{n-1} - \binom{d+n-1+j-1}{n-1}
 \end{matrix} \right),
\end{equation}
depending on wether $j < l$ or $j > l$, and $C_l$ has the form
\begin{equation}
\left( \begin{matrix}
  0 \\
  \binom{d+1+l}{1} - \binom{d+1+l-1-1}{1} \\
  \vdots \\
  \binom{d+k-2+l}{k-2} - \binom{d+k-2+l-1-1}{k-2}\\
  \binom{d+k+l}{k} - \binom{d+k+l-1-1}{k}\\
  \vdots \\
  \binom{d+n-1+l}{n-1} - \binom{d+n-1+l-1-1}{n-1}
 \end{matrix} \right).
\end{equation}
As in the proof of the proposition \ref{propdet}, we use the Pascal's formula $\binom{p+1}{q+1} - \binom{p}{q+1} = \binom{p}{q}$ to simplify the expression of the column $C_j$, $j \neq l$ and $1$, and we obtain 
\begin{equation}
 \left( \begin{matrix}
  0 \\
  \binom{d+0+j-1}{0} \\
  \vdots \\
  \binom{d+k-3+j-1}{k-3}\\
  \binom{d+k-1+j-1}{k-1}\\
  \vdots \\
  \binom{d+n-2+j-1}{n-2}
 \end{matrix} \right) \;\;\; \text{or} \;\;\; 
  \left( \begin{matrix}
  0 \\
  \binom{d+0+j}{0} \\
  \vdots \\
  \binom{d+k-3+j}{k-3}\\
  \binom{d+k-1+j}{k-1}\\
  \vdots \\
  \binom{d+n-2+j}{n-2}
 \end{matrix} \right).
\end{equation}
However, for $j=l$, we have to use the formula $\binom{p+2}{q+1} - \binom{p}{q+1} = \binom{p+1}{q} + \binom{p}{q}$, also derived from the Pascal's one. We obtain therefore 
\begin{equation}
\left( \begin{matrix}
  0 \\
  \binom{d+0+l}{0} + \binom{d+0+l-1}{0} \\
  \vdots \\
  \binom{d+k-3+l}{k-3} + \binom{d+k-3+l-1}{k-3}\\
  \binom{d+k-1+l}{k-1} + \binom{d+k-1+l-1}{k-1}\\
  \vdots \\
  \binom{d+n-2+l}{n-2} + \binom{d+n-2+l-1}{n-2}
 \end{matrix} \right).
\end{equation}
Because the first row contains one $1$ at the first position and $0$ at the others positions, we expand $D_{n-1}^{(k,l)}(d)$ along the first row and we obtain a determinant of size $n-2$ whose columns $C_j$, $j \in [\![1,n-2]\!]$, are, respectively for $j < l-1, j = l-1$ and $j > l-1$,
\begin{equation}
 \left( \begin{matrix}
  \binom{d+1+0+j-1}{0} \\
  \vdots \\
  \binom{d+1+k-2+j-2}{k-3}\\
  \binom{d+1+k-1+j-1}{k-1}\\
  \vdots \\
  \binom{d+1+n-2+j-1}{n-2}
 \end{matrix} \right), \;
 \left( \begin{matrix}
  \binom{d+1+0+l-1}{0} + \binom{d+1+0+l-1-1}{0} \\
  \vdots \\
  \binom{d+1+k-3+l-1}{k-3} + \binom{d+1+k-3+l-1-1}{k-3}\\
  \binom{d+1+k-1+l-1}{k-1} + \binom{d+1+k-1+l-1-1}{k-1}\\
  \vdots \\
  \binom{d+1+n-2+l-1}{n-2} + \binom{d+1+n-2+l-1-1}{n-2}
 \end{matrix} \right), \; 
 \left( \begin{matrix}
  \binom{d+1+0+j}{0} \\
  \vdots \\
  \binom{d+1+k-3+j}{k-3}\\
  \binom{d+1+k-1+j}{k-1}\\
  \vdots \\
  \binom{d+1+n-2+j}{n-2}
 \end{matrix} \right).
\end{equation}
With the rewriting of the columns, we remark that $d$ has been increased by $1$ and the cut between rows takes place now between rows $k-2$ and $k-1$. Moreover, we can separate this new determinant in two ones, because of the sum in $C_{l-1}$, but the cut between columns takes place at two different positions : in the left term, it occurs between $C_{l-2}$ and $C_{l-1}$ because it has the same form that $C_j$ for $j > l-1$ and in the right term, it occurs between $C_{l-1}$ and $C_{l}$ because it has the same form that $C_j$ for $j < l-1$. We obtain finally the wanted inductive relation 
\begin{equation}
 D_{n-1}^{(k,l)}(d) = D_{n-2}^{(k-1,l-1)}(d+1) + D_{n-2}^{(k-1,l)}(d+1).
\end{equation}

Finally, we consider now the cases $l = 1$ and $l = n$. Indeed, for both cases, there is no cut between columns and we only have to use the Pascal's formula $\binom{p+1}{q+1} - \binom{p}{q+1} = \binom{p}{q}$ to simplify the expression of $D_{n-1}^{(k,l)}(d)$. So we obtain 
\begin{equation}
 D_{n-1}^{(k,1)}(d) = D_{n-2}^{(k-1,1)}(d+1) \;\; \text{and} \;\; D_{n-1}^{(k,n)}(d) = D_{n-2}^{(k-1,n-1)}(d+1).
\end{equation}
And, according to the extended definition \ref{defDnkletendu}, $D_{n-2}^{(k-1,0)}(d+1) = D_{n-2}^{(k-1,n)}(d+1) = 0$. So the inductive relation is satisfied again, which concludes the proof.
\end{proof}

\begin{prop}
 For $d \in \NN$, $n \in \NN^*$, $k \in [\![1,n]\!]$ and $l \in \mathbb{Z}$,
 \begin{equation}\label{DnklenfonctionDn1l}
  D_{n-1}^{(k,l)}(d) = \sum_{p=0}^{k-1} \binom{k-1}{p} D_{n-k}^{(1,l-p)}(d+k-1).
 \end{equation}
\end{prop}

\begin{remarks}
 We consider only the case $k \leqslant n$. Else, we will make appear some determinant of negative size, which is not well defined.
 
 This relation will give us an explicit formula for $D_{n-1}^{(k,l)}(d)$ when we will have $D_{n-1}^{(1,l)}(d)$.
\end{remarks}

\begin{proof}
 Let be $n \in \NN^*$. For $n = 1$, the equality is automatically satisfied because the sum contains only one term. Now, for $n >1$, we will prove by induction over $k \in [\![1,n]\!]$ the following statement
 \begin{equation}
  \forall d \in \NN, l \in \mathbb{Z}, \;\; D_{n-1}^{(k,l)}(d) = \sum_{p=0}^{k-1} \binom{k-1}{p} D_{n-k}^{(1,l-p)}(d+k-1).
 \end{equation}
For $k=1$, the equality is automatically satisfied and the base case is proved. 
Now, we suppose that the statement is true for one $k \in \NN^*$ and we consider $D_{n-1}^{(k+1,l)}(d)$. By the inductive relation \eqref{relrecDnkl},
\begin{equation}
 D_{n-1}^{(k+1,l)}(d) = D_{n-2}^{(k,l-1)}(d+1) + D_{n-2}^{(k,l)}(d+1)
\end{equation}
and using the inductive statement twice,
\begin{equation}
 D_{n-1}^{(k+1,l)}(d) = \sum_{p=0}^{k-1} \binom{k-1}{p} D_{n-1-k}^{(1,l-1-p)}(d+k-1) + \sum_{p=0}^{k-1} \binom{k-1}{p} D_{n-1-k}^{(1,l-p)}(d+k-1).
\end{equation}
By changing the index of the first sum ($p'=p+1$), we can rewrite the previous expression 
\begin{align}
 D_{n-1}^{(k+1,l)}(d) &= \sum_{p=1}^{k} \binom{k-1}{p-1} D_{n-1-k}^{(1,l-p)}(d+k-1) + \sum_{p=0}^{k-1} \binom{k-1}{p} D_{n-1-k}^{(1,l-p)}(d+k-1) \\
 &= D_{n-1-k}^{(1,l-k)}(d+k-1) + \sum_{p=1}^{k-1} \left[ \binom{k-1}{p-1} + \binom{k-1}{p} \right] D_{n-1-k}^{(1,l-p)}(d+k-1) \notag \\ 
 &+ D_{n-1-k}^{(1,l)}(d+k-1) \\
 &=D_{n-1-k}^{(1,l-k)}(d+k-1) + \sum_{p=1}^{k-1} \binom{k}{p} D_{n-1-k}^{(1,l-p)}(d+k-1) + D_{n-1-k}^{(1,l)}(d+k-1)
\end{align}
using the Pascal's formula. We obtain finally the inductive statement for $k+1$
\begin{equation}
 D_{n-1}^{(k+1,l)}(d) = \sum_{p=0}^{k} \binom{k}{p} D_{n-1-k}^{(1,l-p)}(d+k-1)
\end{equation}
which concludes the proof.
\end{proof}

\begin{prop}
  For $d \in \NN$, $n \in \NN^*$, $k,l \in [\![1,n]\!]$, $D_{n-1}^{(k,l)}(d)$ is defined uniquely with the equality 
 \begin{equation}
  \forall k,j \in [\![1,n]\!], \;\; \sum_{l=1}^n (-1)^{k+l} D_{n-1}^{(k,l)}(d) \binom{d+l-1+j-1}{j-1} = \delta_{kj}.
 \end{equation}
\end{prop}

\begin{proof}
 Because $D_{n-1}^{(k,l)}(d)$ is the $(k,l)$ minor of $M_n(d)$, $(-1)^{k+l} D_{n-1}^{(k,l)}(d)$ is one coefficient of the comatrix of $M_n(d)$, $\mathrm{Com}(M_n(d))$. 
 Moreover, we know that 
 \begin{align}
  (M_n(d)^T)^{-1} &= \frac{1}{\mathrm{det}(M_n(d)^T)} (\mathrm{Com}(M_n(d)^T))^T = \frac{1}{D_n(d)} \mathrm{Com}(M_n(d)) = \mathrm{Com}(M_n(d))
 \end{align}
using the proposition \ref{propdet}. So $\mathrm{Com}(M_n(d))$ can be defined uniquely as the inverse matrix of $M_n(d)^T$ and their coefficient can be defined uniquely also with the relations
\begin{equation}
  \forall k,j \in [\![1,n]\!], \;\; \sum_{l=1}^n (-1)^{k+l} D_{n-1}^{(k,l)}(d) \binom{d+l-1+j-1}{j-1} = \delta_{kj}.
 \end{equation}
\end{proof}

\begin{prop}
 For $d \in \NN$, $n \in \NN^*$, $k \in \NN^*$ and $l \in \mathbb{Z}$, $D_{n-1}^{(k,l)}(d)$ is given by 
 \begin{itemize}
  \item $D_{n-1}^{(1,l)}(d) = \frac{\mathds{1}_{1\leqslant l \leqslant n}}{(l-1)!(n-l)!} \prod_{1 \leqslant q \leqslant n, \, q \neq l} (d + q)$,
  \item $D_0^{(k,l)}(d) = \mathds{1}_{k=l=1}$ 
  \item and, for $n > 1$, $D_{n-1}^{(k,l)}(d) = D_{n-2}^{(k-1,l-1)}(d+1) + D_{n-2}^{(k-1,l)}(d+1)$
 \end{itemize}
 where $\mathds{1}$ is the usual indicator function.
\end{prop}

\begin{proof}
 We will prove by induction over $k \in \NN^*$, using the previous expressions, the following statement
 \begin{equation}
  \forall d \in \NN, \, n \in \NN^*, \, j \in [\![1,n]\!], \; \;\; \sum_{l=1}^n (-1)^{k+l} D_{n-1}^{(k,l)}(d) \binom{d+l-1+j-1}{j-1} = \delta_{kj}.
 \end{equation}
 
For $k = 1$, let be $d \in \NN$ and $n \in \NN^*$. We have to prove 
\begin{equation}
 \sum_{l=1}^n \left( \frac{(-1)^{1+l}}{(l-1)!(n-l)!} \underset{q\neq l}{\prod_{q = 1}^n} (d + q) \right) = 1
\end{equation}
and, if $n > 1$,
\begin{equation}
 \forall j \in [\![2,n]\!], \;\; \sum_{l=1}^n \left( \frac{(-1)^{1+l}}{(l-1)!(n-l)!} \underset{q\neq l}{\prod_{q = 1}^n} (d + q) \binom{d+l-1+j-1}{j-1} \right) = 0.
\end{equation}
For the first equation, it is equivalent to prove that 
\begin{equation}
  \prod_{q=1}^n \frac{1}{d + q} = \sum_{l=1}^n \left( \frac{(-1)^{1+l}}{(l-1)!(n-l)!} \frac{1}{d + l} \right)
\end{equation}
which corresponds to the partial fraction expansion of the rational fraction $\prod_{q=1}^n \frac{1}{d + q}$. We will compute this expansion. According to the theory, the partial fraction expansion of this rational fraction has the form 
\begin{equation}
 \prod_{q=1}^n \frac{1}{d + q} = \sum_{l=1}^n \frac{\gamma_l}{d + l} 
\end{equation}
where
\begin{equation}
 \gamma_l = \lim\limits_{d \rightarrow -l} \underset{q\neq l}{\prod_{q = 1}^n} \frac{1}{d + q} = \prod_{q=1}^{l-1} \frac{-1}{l-q} \prod_{q=l+1}^n \frac{1}{q-l} = \frac{(-1)^{l-1}}{(l-1)!(n-l)!}
\end{equation}
which proves the first equation because $(-1)^{l-1} = (-1)^{l+1}$. Now, for the second equation (and $n > 1$), it is equivalent to prove that 
\begin{equation}
 \forall j \in [\![2,n]\!], \;\; \sum_{l=1}^n \left( \frac{(-1)^{1+l}}{(l-1)!(n-l)!} \frac{1}{d+l} \binom{d+l-1+j-1}{j-1} \right) = 0
\end{equation}
or
\begin{equation}
 \forall j \in [\![2,n]\!], \;\; \sum_{l=1}^n \left( \frac{(-1)^{1+l}}{(l-1)!(n-l)!} \prod_{q=1}^{j-2} (d + l + q) \right) = 0
\end{equation}
developing the binomial coefficient. Multiplying this sum by $(n-1)!$ and making the index change $l'=l-1$, it is equivalent to prove that
\begin{equation}
 \forall j \in [\![2,n]\!], \;\; \sum_{l=0}^{n-1} \left( (-1)^{l} \binom{n-1}{l} \prod_{q=1}^{j-2} (d + l + 1 + q) \right) = 0.
\end{equation}
However, for $j \in [\![2,n]\!]$,
\begin{equation}
 \prod_{q=1}^{j-2} (d + l + 1 + q) = \frac{\mathrm{d}}{\mathrm{d}x^{j-2}} \left.\left( x^{d+l+j-1} \right)\right\vert_{x=1}
\end{equation}
so 
\begin{align}
 \sum_{l=0}^{n-1} \left( (-1)^{l} \binom{n-1}{l} \prod_{q=1}^{j-2} (d + l + 1 + q) \right) &= \frac{\mathrm{d}}{\mathrm{d}x^{j-2}} \left.\left( x^{d+j-1} \sum_{l=0}^{n-1} \binom{n-1}{l} (-x)^l \right)\right\vert_{x=1} \\
 &= \frac{\mathrm{d}}{\mathrm{d}x^{j-2}} \left.\left( x^{d+j-1} (1-x)^{n-1} \right)\right\vert_{x=1}.
\end{align}
$1$ is a order $n-1$ root of the polynom $X^{d + j-1}(1-X)^{n-1}$. But $0 \leqslant j-2 < n-1$ because $j \in [\![2,n]\!]$. So $1$ is still a root of its $(j-2)^{\text{th}}$ derivative and the second equation is true. The base case is now proved.

We consider now that the statement is true for one $k \in \NN^*$. Let be $d \in \NN$ and $n \in \NN^*$. If $n = 1$, because the case $k=1$ has been proved just before, the equality is automatically verified for $k+1$, for any $k\in \NN^*$. We suppose now that $n > 1$ and we consider the sum for $k+1$. For $j \in [\![1,n]\!]$, using the inductive relation,
\begin{align}
   \sum_{l=1}^n (-1)^{k+1+l} &D_{n-1}^{(k+1,l)}(d) \binom{d+l-1+j-1}{j-1} \notag \\
   &= \sum_{l=1}^n (-1)^{k+1+l} \left( D_{n-2}^{(k,l-1)}(d+1) + D_{n-2}^{(k,l)}(d+1) \right) \binom{d+l-1+j-1}{j-1}\\
   &= \sum_{l=2}^n (-1)^{k+1+l} D_{n-2}^{(k,l-1)}(d+1) \binom{d+l-1+j-1}{j-1} \notag \\ 
   &+ \sum_{l=1}^{n-1} (-1)^{k+1+l} D_{n-2}^{(k,l)}(d+1) \binom{d+l-1+j-1}{j-1}
\end{align}
because for all $d \in \NN$, $D_{n-2}^{(i,n)}(d) = D_{n-2}^{(i,0)}(d) = 0$. So, making a change of index in the first sum, we have
\begin{align}
 \sum_{l=1}^n (-1)^{k+1+l} &D_{n-1}^{(k+1,l)}(d) \binom{d+l-1+j-1}{j-1} \notag \\
 &= \sum_{l=1}^{n-1} (-1)^{k+l} D_{n-2}^{(k,l)}(d+1) \left[ \binom{d+l+j-1}{j-1} - \binom{d+l-1+j-1}{j-1} \right].    
\end{align}
If $j=1$, this sum is therefore equal to $0$, that is ww want to prove because $k+1 >1$. Else, 
\begin{align}
 \sum_{l=1}^n (-1)^{k+1+l} &D_{n-1}^{(k+1,l)}(d) \binom{d+l-1+j-1}{j-1} \notag \\
 &= \sum_{l=1}^{n-1} (-1)^{k+l} D_{n-2}^{(k,l)}(d+1) \binom{d+l+j-2}{j-2} 
\end{align}
using the Pascal's formula. Remarking that $\binom{d+l+j-2}{j-2}  = \binom{d+1+l-1+j-2}{j-2}$, we obtain, using the inductive statement for $k$, 
\begin{align}
 \sum_{l=1}^n &(-1)^{k+1+l} D_{n-1}^{(k+1,l)}(d) \binom{d+l-1+j-1}{j-1} \notag \\
 &= \sum_{l=1}^{n-1} (-1)^{k+l} D_{n-2}^{(k,l)}(d+1) \binom{d+1+l-1+j-2}{j-2} = \delta_{k,j-1} = \delta_{k+1,j}
\end{align}
which concludes the proof.
\end{proof} 

Now, we use all the previous propositions to obtain an explicit expression to $D_{n-1}^{(k,l)}(d)$.

\begin{prop}
 For $d \in \NN$, $n \in \NN^*$, $k \in \NN^*$ and $l \in \mathbb{Z}$, 
 \begin{equation}
  D_{n-1}^{(k,l)}(d) = \frac{\mathds{1}_{1 \leqslant k,l \leqslant n}}{(n-k)!} \prod_{q=k}^n (d+q) \sum_{p = \mathrm{max}(l,k)}^{\mathrm{min}(k+l-1,n)} \binom{k-1}{p-l} \binom{n-k}{p-k}  \frac{1}{d+p}.
 \end{equation}
\end{prop}

\begin{proof}
 We consider directly $k,l \in [\![1,n]\!]$. Else, we already know that $D_{n-1}^{(k,l)}(d) = 0$. Using the equation \eqref{DnklenfonctionDn1l} and the expression of $D_{n-1}^{(1,l)}(d)$, we obtain
 \begin{equation}
  D_{n-1}^{(k,l)}(d) = \sum_{p=0}^{k-1} \binom{k-1}{p} \frac{\mathds{1}_{1 \leqslant l-p \leqslant n-k+1}}{(l-p-1)!(n-k+1-l+p)!} \underset{q \neq l-p}{\prod_{q = 1}^{n-k+1}} (d + k -1 + q).
 \end{equation}
 Making successively changes of indices $q'=q+k-1$, $p'=k-1-p$ and then $p' = l+p$ (with a rewriting before), we obtain
 \begin{align}
  D_{n-1}^{(k,l)}(d) &= \sum_{p=0}^{k-1} \binom{k-1}{p} \frac{\mathds{1}_{1 \leqslant l-p \leqslant n-k+1}}{(l-p-1)!(n-k+1-l+p)!} \underset{q \neq l-p+k-1}{\prod_{q = k}^{n}} (d + q) \\
  &= \sum_{p=0}^{k-1} \binom{k-1}{p} \frac{\mathds{1}_{1 \leqslant l+p-k+1 \leqslant n-k+1}}{(l+p-k)!(n-l-p)!} \underset{q \neq l+p}{\prod_{q = k}^{n}} (d + q) \\
  &= \sum_{p=0}^{k-1} \binom{k-1}{p} \frac{\mathds{1}_{k-l \leqslant p \leqslant n-l}}{(p-(k-l))!(n-l-p)!} \underset{q \neq l+p}{\prod_{q = k}^{n}} (d + q) \\
  &= \sum_{p=l}^{k+l-1} \binom{k-1}{p-l} \frac{\mathds{1}_{k \leqslant p \leqslant n}}{(p-k)!(n-p)!} \underset{q \neq p}{\prod_{q = k}^{n}} (d + q).
 \end{align}
Injecting the bounds of the indicator function with its of the sum and rewriting again the expression, we obtain the wanted expression.
\end{proof}

\begin{examples}
 In our computations, we have to obtain the expression of $D_{n-1}^{(k,1)}(d)$ for some $d$ and $k$. The sum contains one term, which is $\frac{1}{d+k}$ and it appears that its expression is very simple :
 \begin{align}
  D_{n-1}^{(k,1)}(d) = \mathds{1}_{1 \leqslant k \leqslant n}  \binom{d+n}{d+k}.
 \end{align}
 Therefore,
 \begin{equation}\label{110}
  D_{n-1}^{(1,1)}(0) = \binom{n}{1} = n,
 \end{equation}
 \begin{equation}\label{210}
  D_{n-1}^{(2,1)}(0) = \binom{n}{2} = \frac{n(n-1)}{2},
 \end{equation}
 \begin{equation}\label{111}
   D_{n-1}^{(1,1)}(1) = \binom{n+1}{2} = \frac{n(n+1)}{2}.
 \end{equation}
\end{examples}

\subsection{Detailed computations}\label{derivativeZ}

We consider the equation \eqref{Zhomofin} where we factorize the term $\left(1+\frac{x}{c}\right)^{d+i+j-1}$ out of the determinant
\begin{equation}
 Z_{nm}(x\vert B \vert \hat C) = Z^0_{nm}(B \vert \hat C) \left(1+\frac{x}{c}\right)^{nm}\underset{1\leqslant j \leqslant \min(n,m)}{\mathrm{det}} \left(C_j(x)\right)
\end{equation}
with 
\begin{equation}
 C_j(x) = \left( \binom{d +i+j-2}{i-1} \left(1-\beta\left(\frac{x}{x+c}\right)^{d+i+j-1} \right)\right)_{1\leqslant i \leqslant \min(n,m)}.
\end{equation}
So 
\begin{align}
 \frac{Z_{nm}'(0\vert B \vert \hat C)}{Z^0_{nm}(B \vert \hat C)} &= \frac{nm}{c} \underset{1\leqslant j \leqslant \min(n,m)}{\mathrm{det}} \left(C_j(0)\right) \notag \\ 
 &+ \sum_{j=1}^{\min(n,m)} \mathrm{det}\left(C_1(0), \dots, C_{j-1}(0), C_j'(0), C_{j+1}(0), \dots , C_{\min(n,m)}(0)\right)
\end{align}
and
\begin{align}
 \frac{Z_{nm}''(0\vert B \vert \hat C)}{Z^0_{nm}(B \vert \hat C)} &= \frac{nm(nm-1)}{c^2} \underset{1\leqslant j \leqslant \min(n,m)}{\mathrm{det}} \left(C_j(0)\right) \notag \\
 &+ 2\frac{nm}{c} \sum_{j=1}^{\min(n,m)} \mathrm{det}\left(C_1(0), \dots, C_{j-1}(0), C_j'(0), C_{j+1}(0), \dots , C_{\min(n,m)}(0)\right) \notag \\
 &+ 2 \sum_{1 \leqslant j_1 < j_2 \leqslant\min(n,m)} \mathrm{det}\left(C_1(0), \dots, C_{j_1}'(0), \dots, C_{j_2}'(0), \dots , C_{\min(n,m)}(0)\right) \notag \\
 &+ \sum_{j=1}^{\min(n,m)} \mathrm{det}\left(C_1(0), \dots, C_{j-1}(0), C_j''(0), C_{j+1}(0), \dots , C_{\min(n,m)}(0)\right)
\end{align}
with 
\begin{align}
 C_j'(x) = \left( \binom{d +i+j-2}{i-1} (-\beta) (d +i+j-1) \frac{cx^{d +i+j-2}}{(x+c)^{d +i+j}} \right)_{1\leqslant i \leqslant \min(n,m)}
\end{align}
and 
\begin{align}
 C_j''(x) &= \left( \binom{d +i+j-2}{i-1} \right.   (-\beta) (d +i+j-1)c \notag \\
 & \times \left. \left[ (d +i+j-2) \frac{x^{d +i+j-3}}{(x+c)^{d +i+j}}  - (d +i+j) \frac{x^{d +i+j-2}}{(x+c)^{d +i+j+1}} \right] \right)_{1\leqslant i \leqslant \min(n,m)}.
\end{align}
Using the proposition \ref{propdet}, we know that $\underset{1\leqslant j \leqslant \min(n,m)}{\mathrm{det}} \left(C_j(0)\right) = 1$ for all $n,m$. However, the results of both previous computation are different according to the value of $d$. 

 \paragraph{For $d > 1$} For all $j \in [\![1,\min(n,m)]\!]$, $C_j'(0) = C_j''(0) = 0$. So 
\begin{equation}
 \frac{Z_{nm}'(0\vert B \vert \hat C)}{Z^0_{nm}(B \vert \hat C)} = \frac{nm}{c}, 
\end{equation}
\begin{equation}
 \frac{Z_{nm}''(0\vert B \vert \hat C)}{Z^0_{nm}(B \vert \hat C)} = \frac{nm(nm-1)}{c^2}
\end{equation}
and 
\begin{equation}
 \frac{Z_{nm}''(0\vert B \vert \hat C)}{Z^0_{nm}(B \vert \hat C)} - \left( \frac{Z_{nm}'(0\vert B \vert \hat C)}{Z^0_{nm}(B \vert \hat C)}\right)^2 = -\frac{nm}{c^2}.
\end{equation}
Recalling that
\begin{equation}
 \frac{Z_{nm}'(0\vert B \vert \hat C)}{Z^0_{nm}(B\vert \hat C)} = -nm\frac{F'(0)}{k_BT} + o(nm)
\end{equation}
and
\begin{equation}
 \frac{Z_{nm}''(0\vert B \vert \hat C)}{Z^0_{nm}(B \vert \hat C)} - \left(\frac{Z_{nm}'(0\vert B \vert \hat C)}{Z^0_{nm}(B \vert \hat C)}\right)^2 = -nm\frac{F''(0)}{k_BT} + o(nm),
\end{equation}
we obtain therefore 
\begin{equation}
\frac{F'(0)}{k_BT} = -\frac{1}{c} \qquad \text{and} \qquad \frac{F''(0)}{k_BT} = \frac{1}{c^2}.
\end{equation}
So $\alpha = 0$.

\paragraph{For $d = 1$} For all $j \in [\![2,\min(n,m)]\!]$, $C_1'(0) = C_j'(0) = C_j''(0) = 0$ but 
\begin{equation}
 C_1''(0) = \left( \begin{matrix} -\frac{2\beta}{c^2} \\ 0 \\ \vdots \\ 0 \end{matrix} \right).
\end{equation}
So 
\begin{equation}
 \frac{Z_{nm}'(0\vert B \vert \hat C)}{Z^0_{nm}(B \vert \hat C)} = \frac{nm}{c} \quad \text{which implies} \quad \frac{F'(0)}{k_BT} = -\frac{1}{c}.
\end{equation}
And
\begin{equation}
 \frac{Z_{nm}''(0\vert B \vert \hat C)}{Z^0_{nm}(B \vert \hat C)} = \frac{nm(nm-1)}{c^2} + \mathrm{det}\left(C_1''(0), C_2(0), \dots, C_{\min(n,m)}(0)\right).
\end{equation}
We expand the determinant along the first column 
\begin{align}
 \mathrm{det}\left(C_1''(0), C_2(0), \dots, C_{\min(n,m)}(0)\right) &= -\frac{2\beta}{c^2} \underset{1 \leqslant i,j \leqslant \min(n,m)-1}{\mathrm{det}} \left( \binom{i+j+1}{i} \right) \\
 &= -\frac{\beta}{c^2} nm.
\end{align}
using the formula \eqref{111} for $D_{\min(n,m)-1}^{(1,1)}(1)$. Therefore, 
\begin{equation}
 \frac{Z_{nm}''(0\vert B \vert \hat C)}{Z^0_{nm}(B \vert \hat C)} - \left( \frac{Z_{nm}'(0\vert B \vert \hat C)}{Z^0_{nm}(B \vert \hat C)} \right)^2 = -nm \frac{1+\beta}{c^2}.
\end{equation}
We obtain 
\begin{equation}
 \frac{F''(0)}{k_BT} = \frac{1+\beta}{c^2}
\end{equation}
so $\alpha^2 = \frac{3\beta}{c^2}$.

\paragraph{For $d = 0$ ($n = m$)} For all $j \in [\![3,n]\!]$, $C_2'(0) = C_j'(0) = C_j''(0) = 0$ but
\begin{equation}
 C_1'(0) = \left( \begin{matrix} -\frac{\beta}{c} \\ 0 \\ \vdots \\ 0 \end{matrix} \right), \qquad
 C_1''(0) = \left( \begin{matrix} \frac{2\beta}{c^2} \\ -\frac{2\beta}{c^2} \\ 0 \\ \vdots \\ 0 \end{matrix} \right) \qquad \text{and} \qquad 
 C_2''(0) = \left( \begin{matrix} -\frac{2\beta}{c^2} \\ 0 \\ \vdots \\ 0 \end{matrix} \right) .
\end{equation}
So 
\begin{equation}
 \frac{Z_{nn}'(0\vert B \vert \hat C)}{Z^0_{nn}(B \vert \hat C)} = \frac{n^2}{c} + \mathrm{det}\left(C_1'(0), C_2(0), \dots, C_n(0)\right)
\end{equation}
and
\begin{align}
 \frac{Z_{nn}''(0\vert B \vert \hat C)}{Z^0_{nn}(B \vert \hat C)} &= \frac{n^2(n^2-1)}{c^2} +\frac{2n^2}{c} \mathrm{det}\left(C_1'(0), C_2(0), \dots, C_n(0)\right) \notag \\
 &+ \mathrm{det}\left(C_1''(0), C_2(0), \dots, C_n(0)\right) + \mathrm{det}\left(C_1(0), C_2''(0), C_3(0), \dots, C_n(0)\right).
\end{align}
We expand the determinants along the first or second column
\begin{equation}
 \mathrm{det}\left(C_1'(0), C_2(0), \dots, C_n(0)\right) = -\frac{\beta}{c}\underset{1\leqslant i,j \leqslant n-1}{\mathrm{det}}\left( \binom{i+j}{i} \right) = -\frac{\beta}{c}n
\end{equation}
\begin{align}
 \mathrm{det}\left(C_1''(0), C_2(0), \dots, C_n(0)\right) &= \frac{2\beta}{c^2} \left( \underset{1\leqslant i,j \leqslant n-1}{\mathrm{det}}\left( \binom{i+j}{i} \right) \right. \notag \\
 &+ \left. \underset{1\leqslant i,j \leqslant n-1}{\mathrm{det}}\left( \left\{
 \begin{array}{cc}
    1 & \text{if} \; i = 1 \\
    \binom{i+j}{i} & \text{if} \; i \geqslant 2
    \end{array} \right.
  \right) \right) \\
  &=  \frac{2\beta}{c^2} \left( n + \frac{n(n-1)}{2} \right)
\end{align}
\begin{align}
 \mathrm{det}\left(C_1(0), C_2''(0), C_3(0), \dots, C_n(0)\right) &=  \frac{2\beta}{c^2} \underset{1\leqslant i,j \leqslant n-1}{\mathrm{det}}\left( \left\{\begin{array}{cc}
    1 & \text{if} \; j = 1 \\
    \binom{i+j}{i} & \text{if} \; j \geqslant 2                                          
    \end{array} \right.
  \right) \\
  &= \frac{2\beta}{c^2} \frac{n(n-1)}{2}
\end{align}
using the formulae \eqref{110}-\eqref{210} for $D_{n-1}^{(1,1)}(0)$ and $D_{n-1}^{(2,1)}(0)$ respectively .
So 
\begin{equation}
 \frac{Z_{nn}'(0\vert B \vert \hat C)}{Z_{nn}(0\vert B \vert \hat C)} = \frac{n^2-\beta n}{c}
\end{equation}
and
\begin{equation}
 \frac{Z_{nn}''(0\vert B \vert \hat C)}{Z_{nn}(0\vert B \vert \hat C)} = \frac{n^2(n^2-1)}{c^2} - \frac{2\beta}{c^2}n^2(n-1).
\end{equation}
The first equation still implies 
\begin{equation}
 \frac{F'(0)}{k_BT} = -\frac{1}{c^2}
\end{equation}
because $\beta$ is in the linear term. However, with the second equation, we obtain
\begin{align}
 -n^2\frac{F''(0)}{k_BT} + o(n^2) = \frac{n^2(n^2-1)}{c^2} - \frac{2\beta}{c^2}n^2(n-1) - \frac{(n^2-n\beta)^2}{c^2} = -n^2\frac{(\beta-1)^2}{c^2}
\end{align}
that implies $\alpha^2 = \frac{3\beta(\beta-2)}{c^2}$.

\section{Some useful mathematical tools}\label{mathematicaltools}

\subsection{Basics on elementary and complete symmetric functions}\label{symmfunc}

Informations of this subsection come from \cite{MacDonald1995}. The proofs are inside the book or are easy to write.

\begin{Def}\label{defsymmfunc}
 We consider $x_1, \dots x_N$ $N$ independent variables and $\bar{x}$ the set of this variables. The $r$-th elementary symmetric function $e_r(\bar{x})$ is the sum of all products of $r$ distinct variables $x_i$. Therefore, $e_0(\bar{x}) = 1$ and for $r \in \NN^*$,
 \begin{equation}
  e_r(\bar{x}) = \sum_{1 \leqslant i_1 < \cdots < i_r \leqslant N} x_{i_1} \dots x_{i_r} \;\; \text{if} \; r \leqslant N \;\; \text{and} \;\; e_r(\bar{x}) = 0 \;\; \text{else.}
 \end{equation}
 
 The $r$-th complete symmetric function $h_r(\bar{x})$ is the sum of all monomials of total degree $r$ in the variables $x_i$. Therefore, $h_0(\bar{x}) = 1$ and for $r \in \NN^*$,
 \begin{equation}
  h_r(\bar{x}) = \sum_{1 \leqslant i_1 \leqslant \cdots \leqslant i_r \leqslant N} x_{i_1} \dots x_{i_r}.
 \end{equation}
 \end{Def}

 \begin{prop}
   The generating functions for respectively the $e_r(\bar{x})$ and the $h_r(\bar{x})$ are 
   \begin{equation}\label{genersymmfunc}
    E_{\bar{x}}(t) = \sum_{r = 0}^N e_r(\bar{x}) t^r = \prod_{i=1}^N (1+x_it) \;\;\; \text{and} \;\;\; H_{\bar{x}}(t) = \sum_{r \geqslant 0} h_r(\bar{x}) t^r = \prod_{i=1}^N (1-x_it)^{-1}.
   \end{equation}
   We observe that $H_{\bar{x}}(t)E_{\bar{x}}(-t) = 1$ so that, for all $n \in \NN^*$,
   \begin{equation}\label{identitysymmfunc}
    \sum_{r=0}^n (-1)^r e_r(\bar{x}) h_{n-r}(\bar{x}) = 0.
   \end{equation} 
 \end{prop}
 
 \subsection{The Vandermonde matrix and its inverse}
 
 \begin{Def}
   The Vandermonde matrix $V(\bar{x})$ of the $N$ independent variables $x_1, \dots x_N$ is $V(\bar{x}) = (x_i^{j-1})_{1 \leqslant i,j \leqslant N}$.
 \end{Def}
 
 \begin{prop}
  The determinant of the Vandermonde matrix is 
  \begin{equation}
      \mathrm{det}(V(\bar{x})) = \prod_{1 \leqslant i,j \leqslant N} (x_j-x_i).
  \end{equation}
  It corresponds to $\Delta(\bar{x})$ considering $c=1$. If all the $x_i$ are different, $V(\bar{x})$ is invertible and its inverse matrix is
  \begin{equation}
   V(\bar{x})^{-1} = \left( \frac{(-1)^{N-i} e_{N-i}(\bar{x_j})}{\prod_{k \neq j} (x_j-x_k)} \right)_{1 \leqslant i,j \leqslant N}.
  \end{equation}
  The interested reader can find a proof in \cite{Gourdon2021} for the determinant and in \cite{Janjic2013} for the inverse matrix.
 \end{prop}

 \begin{remark}
  We obtain directly the equation
  \begin{equation}\label{eqVand}
   \forall i,j \in [\![1,N]\!], \;\;\; \sum_{k=1}^N x_k^{j-1} (-1)^{N-i}e_{N-i}(\bar{x}_k) \underset{q \neq k}{\prod_{q=1}^N} \frac{1}{x_k-x_q}  = \delta_{ij}.
  \end{equation}
  with the coefficient $(i,j)$ of the matrix $V(\bar{x})^{-1} V(\bar{x})$.
 \end{remark}

 \subsection{Complex analysis}

\begin{prop}\label{calculsum}
    We consider $F$ a complex rational function $F(z) = P(z)/Q(z)$ where $P$ and $Q$ are complex polynomials with $\mathrm{deg} \, Q > \mathrm{deg} \, P + 1$. We denote $a_1, \dots, a_q $ the poles of $F$ (in other words the zeros of $Q$). We assume that there are only one order pole. Then
    \begin{equation}
        \sum_{k=1}^q \mathrm{Res}(F,a_k) = 0.
    \end{equation}
\end{prop}

\begin{proof}
    We take $R > \mathrm{max} |a_k| $. There exists an open set $U$ of $\CC$ that contains the circle of radius $R$ and center $z = 0$ but not the poles of $F$. $F$ is continuous on $U$ so, according to the theorem 2.3 of the section III §2 of \cite{Lang1999}, 
    \begin{equation}
        \left| \oint_{|z|=R} F(z) \deriv z \right| \leqslant \underset{|z|=R}{\text{max}} |F(z)| \times 2\pi R.
    \end{equation}
    We can factorize $P$ and $Q$ in irreducible elements 
    \begin{equation}
        P(z) = A \prod_{j=1}^p (z-b_j) \;\; \text{with } A\neq 0 \;\; \text{and} \;\; Q(z) = \prod_{k=1}^q (z-a_k).
    \end{equation}
    and using the triangle inequality, for $|z| = R$ we have
    \begin{equation}
        |z-b_j| \leqslant |z|+|b_j| = R + |b_j| \;\;\; \text{and} \;\;\; |z-a_k| \geqslant ||z|-|a_k|| = R - |a_k|. 
    \end{equation}
    So
    \begin{equation}
        \left| \oint_{|z|=R} F(z) \deriv z \right| \leqslant |A| \frac{\prod_{j=1}^p(R + |b_j|)}{\prod_{k=1}^q (R - |a_k|)} \times 2\pi R.
    \end{equation}
    By hypothesis, $q > p+1$ so $q-p-1 > 0$ and the right term goes to zero when $R$ goes to infinity.
    In addition, by the residue theorem, 
    \begin{equation}
        \oint_{|z|=R} F(z) \deriv z = 2i\pi \sum_{k=1}^q \mathrm{Res}(F,a_k)
    \end{equation}
    and the sum is independent of $R$ as long as $R > \mathrm{max} |a_k|$. So $\sum_{k=1}^q \mathrm{Res}(F,a_k) = 0$.
\end{proof}

\newpage

\addcontentsline{toc}{section}{References}

\printbibliography

@Article{Belliard2024,
  author    = {Belliard, Samuel and Pimenta, Rodrigo Alves and Slavnov, Nikita A.},
  journal   = {SciPost Physics},
  title     = {Modified rational six vertex model on the rectangular lattice},
  year      = {2024},
  issn      = {2542-4653},
  month     = jan,
  number    = {1},
  volume    = {16},
  doi       = {10.21468/scipostphys.16.1.009},
  publisher = {Stichting SciPost},
}

@Article{Foda2012,
  author    = {Foda, O. and Wheeler, M.},
  journal   = {Journal of High Energy Physics},
  title     = {Partial domain wall partition functions},
  year      = {2012},
  issn      = {1029-8479},
  month     = jul,
  number    = {7},
  volume    = {2012},
  doi       = {10.1007/jhep07(2012)186},
  publisher = {Springer Science and Business Media LLC},
}

@Article{Izergin1992,
  author    = {Izergin, A G and Coker, D A and Korepin, V E},
  journal   = {Journal of Physics A: Mathematical and General},
  title     = {Determinant formula for the six-vertex model},
  year      = {1992},
  issn      = {1361-6447},
  month     = aug,
  number    = {16},
  pages     = {4315--4334},
  volume    = {25},
  doi       = {10.1088/0305-4470/25/16/010},
  publisher = {IOP Publishing},
}

@Book{Chari2000,
  author    = {Chari, Vyjayanthi and Pressley, Andrew},
  publisher = {Cambridge Univ. Press},
  title     = {A guide to quantum groups},
  year      = {2000},
  address   = {Cambridge},
  edition   = {Reprint},
  isbn      = {0521558840},
  pagetotal = {651},
  ppn_gvk   = {1616060751},
}

@Article{Belliard2018,
  author    = {Belliard, S and Slavnov, N A and Vallet, B},
  journal   = {Journal of Statistical Mechanics: Theory and Experiment},
  title     = {Scalar product of twisted XXX modified Bethe vectors},
  year      = {2018},
  issn      = {1742-5468},
  month     = sep,
  number    = {9},
  pages     = {093103},
  volume    = {2018},
  doi       = {10.1088/1742-5468/aaddac},
  publisher = {IOP Publishing},
}

@Article{Izergin1987,
  author  = {A. G. Izergin},
  journal = {Doklady Akademii Nauk SSSR},
  title   = {Partition function of a six-vertex model in a finite volume},
  year    = {1987},
  month   = nov,
  number  = {2},
  pages   = {331--333},
  volume  = {297},
  url     = {https://www.mathnet.ru/php/archive.phtml?wshow=paper&jrnid=dan&paperid=7902&option_lang=eng},
}

@Book{Baxter1982,
  author    = {Baxter, Rodney J.},
  publisher = {Dover Publications, Inc},
  title     = {Exactly solved models in statistical mechanics},
  year      = {1982},
  address   = {Mineola, N.Y},
  isbn      = {0486462714},
  note      = {"This Dover edition, first published in 2007, is an unabridged republication of the third edition of the work originally published by Academic Press, London, in 1982. A new chapter, "Subsequent Developments," has been prepared by the author for the present edition."},
  series    = {Dover books on physics},
  pagetotal = {498},
  ppn_gvk   = {543966232},
}

@Article{Janjic2013,
  author        = {Janjic, Milan},
  title         = {Some Determinantal Identities},
  year          = {2013},
  month         = feb,
  archiveprefix = {arXiv},
  copyright     = {arXiv.org perpetual, non-exclusive license},
  doi           = {10.48550/ARXIV.1302.2504},
  eprint        = {1302.2504},
  primaryclass  = {math.CO},
  publisher     = {arXiv},
}

@Article{Pronko2019,
  author    = {Pronko, A. G. and Pronko, G. P.},
  journal   = {Journal of Mathematical Sciences},
  title     = {Off-Shell Bethe States and the Six-Vertex Model},
  year      = {2019},
  issn      = {1573-8795},
  month     = sep,
  number    = {5},
  pages     = {742--752},
  volume    = {242},
  doi       = {10.1007/s10958-019-04511-7},
  publisher = {Springer Science and Business Media LLC},
}

@Article{Minin2021,
  author    = {Minin, Mikhail D. and Pronko, Andrei G.},
  journal   = {Symmetry, Integrability and Geometry: Methods and Applications},
  title     = {Boundary One-Point Function of the Rational Six-Vertex Model with Partial Domain Wall Boundary Conditions: Explicit Formulas and Scaling Properties},
  year      = {2021},
  issn      = {1815-0659},
  month     = dec,
  doi       = {10.3842/sigma.2021.111},
  publisher = {SIGMA (Symmetry, Integrability and Geometry: Methods and Application)},
}

@Article{Korepin2000,
  author    = {Korepin, V and Zinn-Justin, P},
  journal   = {Journal of Physics A: Mathematical and General},
  title     = {Thermodynamic limit of the six-vertex model with domain wall boundary conditions},
  year      = {2000},
  issn      = {1361-6447},
  month     = sep,
  number    = {40},
  pages     = {7053--7066},
  volume    = {33},
  doi       = {10.1088/0305-4470/33/40/304},
  publisher = {IOP Publishing},
}

@Article{ZinnJustin2000,
  author    = {Zinn-Justin, P.},
  journal   = {Physical Review E},
  title     = {Six-vertex model with domain wall boundary conditions and one-matrix model},
  year      = {2000},
  issn      = {1095-3787},
  month     = sep,
  number    = {3},
  pages     = {3411--3418},
  volume    = {62},
  doi       = {10.1103/physreve.62.3411},
  publisher = {American Physical Society (APS)},
}

@Book{Korepin1993,
  author    = {Korepin, V E and Izergin, A G and Bogoliubov, N M},
  publisher = {Cambridge University Press},
  title     = {Quantum inverse scattering method and correlation functions},
  year      = {1993},
  address   = {Cambridge},
  isbn      = {9780511628832},
  note      = {Title from publisher's bibliographic system (viewed on 05 Oct 2015)},
  series    = {Cambridge monographs on mathematical physics},
  pagetotal = {1555},
  ppn_gvk   = {883376725},
}

@Article{Korepin1982,
  author    = {Korepin, V. E.},
  journal   = {Communications in Mathematical Physics},
  title     = {Calculation of norms of Bethe wave functions},
  year      = {1982},
  issn      = {1432-0916},
  month     = sep,
  number    = {3},
  pages     = {391--418},
  volume    = {86},
  doi       = {10.1007/bf01212176},
  publisher = {Springer Science and Business Media LLC},
}

@Article{Tsuchiya1998,
  author    = {Tsuchiya, Osamu},
  journal   = {Journal of Mathematical Physics},
  title     = {Determinant formula for the six-vertex model with reflecting end},
  year      = {1998},
  issn      = {1089-7658},
  month     = nov,
  number    = {11},
  pages     = {5946--5951},
  volume    = {39},
  doi       = {10.1063/1.532606},
  publisher = {AIP Publishing},
}

@Article{Kostov2012,
  author    = {Kostov, Ivan},
  journal   = {Physical Review Letters},
  title     = {Classical Limit of the Three-Point Function ofN=4Supersymmetric Yang-Mills Theory from Integrability},
  year      = {2012},
  issn      = {1079-7114},
  month     = jun,
  number    = {26},
  pages     = {261604},
  volume    = {108},
  doi       = {10.1103/physrevlett.108.261604},
  publisher = {American Physical Society (APS)},
}

@Article{Kostov2012a,
  author    = {Kostov, Ivan},
  journal   = {Journal of Physics A: Mathematical and Theoretical},
  title     = {Three-point function of semiclassical states at weak coupling},
  year      = {2012},
  issn      = {1751-8121},
  month     = nov,
  number    = {49},
  pages     = {494018},
  volume    = {45},
  doi       = {10.1088/1751-8113/45/49/494018},
  publisher = {IOP Publishing},
}

@Article{Belliard2015,
  author    = {Belliard, Samuel and Pimenta, Rodrigo A.},
  journal   = {Symmetry, Integrability and Geometry: Methods and Applications},
  title     = {Slavnov and Gaudin-Korepin Formulas for Models without U(1) Symmetry: the Twisted XXX Chain},
  year      = {2015},
  issn      = {1815-0659},
  month     = dec,
  doi       = {10.3842/sigma.2015.099},
  publisher = {SIGMA (Symmetry, Integrability and Geometry: Methods and Application)},
}

@Article{Sogo1993,
  author    = {Sogo, Kiyoshi},
  journal   = {Journal of the Physical Society of Japan},
  title     = {Time-Dependent Orthogonal Polynomials and Theory of Soliton –Applications to Matrix Model, Vertex Model and Level Statistics},
  year      = {1993},
  issn      = {1347-4073},
  month     = jun,
  number    = {6},
  pages     = {1887--1894},
  volume    = {62},
  doi       = {10.1143/jpsj.62.1887},
  publisher = {Physical Society of Japan},
}

@Book{Macdonald1995,
  author    = {Macdonald, Ian G.},
  publisher = {Clarendon Press, Oxford University Press},
  title     = {Symmetric functions and Hall polynomials},
  year      = {1995},
  address   = {New York},
  edition   = {2nd},
  isbn      = {0198534892},
  note      = {Includes index},
  series    = {Oxford mathematical monographs},
  pagetotal = {475},
  ppn_gvk   = {399555633},
}

@Book{Lang1999,
  author    = {Lang, Serge},
  publisher = {Springer New York},
  title     = {Complex Analysis},
  year      = {1999},
  isbn      = {9781475730838},
  doi       = {10.1007/978-1-4757-3083-8},
  issn      = {0072-5285},
  journal   = {Graduate Texts in Mathematics},
}

@Book{Gourdon2021,
  author  = {Gourdon, Xavier},
  editor  = {Ellipses},
  title   = {Algèbre Probabilités},
  year    = {2021},
  edition = {3rd},
  isbn    = {978-2-3400-5676-3},
}

@Article{Gorsky2014,
  author    = {Gorsky, A. and Zabrodin, A. and Zotov, A.},
  journal   = {Journal of High Energy Physics},
  title     = {Spectrum of quantum transfer matrices via classical many-body systems},
  year      = {2014},
  issn      = {1029-8479},
  month     = jan,
  number    = {1},
  volume    = {2014},
  doi       = {10.1007/jhep01(2014)070},
  publisher = {Springer Science and Business Media LLC},
}

@Article{Lyberg2018,
  author    = {Lyberg, I. and Korepin, V. and Ribeiro, G. A. P. and Viti, J.},
  journal   = {Journal of Mathematical Physics},
  title     = {Phase separation in the six-vertex model with a variety of boundary conditions},
  year      = {2018},
  issn      = {1089-7658},
  month     = may,
  number    = {5},
  volume    = {59},
  doi       = {10.1063/1.5018324},
  publisher = {AIP Publishing},
}

@Article{Motegi2024a,
  author        = {Motegi, Kohei and Ohkawa, Ryo},
  title         = {Algebraic formulas and Geometric derivation of Source Identities},
  year          = {2024},
  month         = apr,
  archiveprefix = {arXiv},
  copyright     = {Creative Commons Attribution Share Alike 4.0 International},
  doi           = {10.48550/ARXIV.2404.08141},
  eprint        = {2404.08141},
  primaryclass  = {math.AG},
  publisher     = {arXiv},
}

@Article{Sklyanin1988,
  author    = {Sklyanin, E K},
  journal   = {Journal of Physics A: Mathematical and General},
  title     = {Boundary conditions for integrable quantum systems},
  year      = {1988},
  issn      = {1361-6447},
  month     = may,
  number    = {10},
  pages     = {2375--2389},
  volume    = {21},
  doi       = {10.1088/0305-4470/21/10/015},
  publisher = {IOP Publishing},
}

@Article{Galleas2018,
  author        = {Galleas, W.},
  title         = {New determinants in the 8VSOS model with domain-wall boundaries},
  year          = {2018},
  month         = dec,
  archiveprefix = {arXiv},
  copyright     = {arXiv.org perpetual, non-exclusive license},
  doi           = {10.48550/ARXIV.1812.10781},
  eprint        = {1812.10781},
  primaryclass  = {math-ph},
  publisher     = {arXiv},
}

@Article{Pakuliak2008,
  author        = {Pakuliak, S. and Rubtsov, V. and Silantyev, A.},
  journal       = {Journal of Physics A: Mathematical and Theoretical},
  title         = {SOS model partition function and the elliptic weight functions},
  year          = {2008},
  issn          = {1751-8121},
  month         = jun,
  number        = {29},
  pages         = {295204},
  volume        = {41},
  archiveprefix = {arXiv},
  copyright     = {Assumed arXiv.org perpetual, non-exclusive license to distribute this article for submissions made before January 2004},
  doi           = {10.1088/1751-8113/41/29/295204},
  eprint        = {0802.0195},
  primaryclass  = {math.QA},
  publisher     = {IOP Publishing},
}

@Article{Rosengren2008,
  author        = {Rosengren, Hjalmar},
  journal       = {Adv. Appl. Math. 43 (2009), 137-155},
  title         = {An Izergin-Korepin-type identity for the 8VSOS model, with applications to alternating sign matrices},
  year          = {2008},
  month         = jan,
  archiveprefix = {arXiv},
  copyright     = {Assumed arXiv.org perpetual, non-exclusive license to distribute this article for submissions made before January 2004},
  doi           = {10.48550/ARXIV.0801.1229},
  eprint        = {0801.1229},
  primaryclass  = {math.CO},
  publisher     = {arXiv},
}

\end{document}